\documentclass[acmsmall,screen,nonacm]{acmart} 

\AtBeginDocument{%
  \providecommand\BibTeX{{%
    \normalfont B\kern-0.5em{\scshape i\kern-0.25em b}\kern-0.8em\TeX}}}

\setcopyright{acmlicensed}
\copyrightyear{2024}
\acmYear{2024}
\acmDOI{XXXXXXX.XXXXXXX}

\acmJournal{JACM}
\acmVolume{37}
\acmNumber{4}
\acmArticle{111}
\acmMonth{10}

\acmISBN{978-1-4503-XXXX-X/18/06}

\usepackage{multicol}
\usepackage{multirow}
\usepackage{graphicx}
\usepackage{arydshln}
\usepackage{listings, listings-rust}
\usepackage{xspace}
\usepackage{cleveref}
\crefname{theorem}{Thm.}{Thms.}
\crefname{lemma}{Lem.}{Lemmas}
\crefname{corollary}{Cor.}{Cors.}
\crefname{figure}{Fig.}{Figs.}
\crefname{definition}{Defn.}{Defns.}
\crefname{table}{Tab.}{Tabs.}
\crefname{appendix}{Appendix}{Appendices}
\crefformat{section}{\S#2#1#3}
\crefmultiformat{section}{\S#2#1#3}{ and~\S#2#1#3}{, \S#2#1#3}{ and~\S#2#1#3}
\crefname{example}{Ex.}{Exs.}
\crefname{item}{item}{items}
\crefname{footnote}{footnote}{footnotes}
\crefname{observation}{Obs.}{Obs.}
\crefname{remark}{Remark}{Remarks}
\crefname{proposition}{Prop.}{Props.}
\crefname{equation}{Eqn.}{Eqns.}
\crefname{counterexample}{Counterexample}{Counterexamples}
\crefname{property}{Property}{Properties}
\crefname{algorithm}{Algorithm}{Algorithms}
\usepackage{subcaption}
\usepackage[shortlabels]{enumitem}
\usepackage{commands}
\usepackage{mathpartir}
\newcommand{\from}{:=}

\def\vDash{\vdash}
\def\Dashv{\dashv}

\def\VWt#1#2#3#4{\{#1\}#2 \from #3\{#4\}}
\def\GWt#1#2#3#4{\{#1\}#2 \from #3\{#4\}}

\def\kwd#1{\textbf{\texttt{#1}}}

\AtBeginDocument{%
  \setlength\abovedisplayskip{0.5\abovedisplayskip}%
  \setlength\belowdisplayskip{0.5\belowdisplayskip}%
  \setlength\abovedisplayshortskip{0.5\abovedisplayshortskip}%
  \setlength\belowdisplayshortskip{0.5\belowdisplayshortskip}%
  \setlength\floatsep{0.5\floatsep}%
  \setlength\abovecaptionskip{0.5\abovecaptionskip}%
}
\setlist{leftmargin=*}

\newenvironment{DIFnomarkup}{}{}

\begin{document}

\title{Automatic Linear Resource Bound Analysis for Rust via Prophecy Potentials}

\author{Qihao Lian}
\affiliation{%
  \department{Key Laboratory of High Confidence Software Technologies (Peking University), Ministry of Education; School of Computer Science}
  \institution{Peking University}
  \country{China}
}
\email{mepy@stu.pku.edu.cn}
\author{Di Wang}
\affiliation{%
  \department{Key Laboratory of High Confidence Software Technologies (Peking University), Ministry of Education; School of Computer Science}
  \institution{Peking University}
  \country{China}
}
\email{wangdi95@pku.edu.cn}

\begin{abstract}

Rust has become a popular system programming language that strikes a balance between memory safety and performance.
Rust's type system ensures the safety of low-level memory controls; however, a well-typed Rust program is not guaranteed to enjoy high performance.
This article studies static analysis for resource consumption of Rust programs, aiming at understanding the performance of Rust programs.
Although there have been tons of studies on static resource analysis, exploiting Rust's memory safety---especially the borrow mechanisms and their properties---to aid resource-bound analysis, remains unexplored. 

This article presents \rarust{}, a type-based linear resource-bound analysis for well-typed Rust programs.
\rarust{} follows the methodology of automatic amortized resource analysis (AARA) to build a resource-aware type system.
To support Rust's borrow mechanisms, including shared and mutable borrows, \rarust{} introduces shared and novel prophecy potentials to reason about borrows compositionally.
To prove the soundness of \rarust{}, this article proposes Resource-Aware Borrow Calculus (RABC) as a variant of recently proposed Low-Level Borrow Calculus (LLBC).
The experimental evaluation of a prototype implementation of \rarust{} demonstrates that \rarust{} is capable of inferring symbolic linear resource bounds for Rust programs featuring shared and mutable borrows, reborrows, heap-allocated data structures, loops, and recursion. 

    

\end{abstract}

\maketitle

\section{Introduction} \label{section:introduction}


Rust~\cite{sigada:MK14} has become a popular programming language that
aims to bridge the gap between fine-grained memory control and high-level memory safety,
especially in the context of system programming.
Rust introduces a complex type system that features linear types and borrow mechanisms,
enabling system engineers to write high-performance memory-manipulating code without
needing garbage collection or compromising memory safety.
The popularity of Rust indicates that system engineers still care about \emph{resource
consumption}---such as time, memory, and energy usage---of their programs, in order
to achieve high or predictable performance.
However, using Rust does \emph{not} directly mean achieving high performance;
in fact, there exists a book~\cite{misc:RPB20} that collects
techniques to improve the performance of Rust programs, including general and
Rust-specific techniques. Additionally, predicting the performance of a Rust program is often difficult to do manually.
Therefore, understanding the resource consumption---as a general topic in software
engineering---continues to be a crucial issue in Rust programming.



In this paper, we study how to \emph{automatically} and \emph{statically} reason
about the resource consumption of Rust programs via a type system.
In particular, we want to understand how
the reasoning process could benefit from Rust's advanced type system.
Recent years have seen many studies on (i) static analysis for resource consumption,
and (ii) static analysis of Rust programs.
We highlight some mostly-related \emph{type-based} approaches below; other related work will be 
discussed in \cref{sec:related}.


Automatic amortized resource analysis (AARA)---initially proposed by~\citet{AARA-Linear}---has
been a state-of-the-art framework for statically inferring symbolic resource bounds
of various kinds of programs~\cite{JMSCS:HJ22}.
The inferred bounds are \emph{symbolic} in the sense that they are interpreted as functions
of a program's input and output.
Despite the fact that most AARA approaches are focused on analyzing pure functional programs,
there are multiple exceptions related to memory control:
\citet{ESOP:Atkey10} studied the integration of AARA and separation logic,
\citet{PLDI:CHS15} developed C4B, which can analyze C programs with arrays,
and \citet{FSCD:LH17} proposed a method to analyze arrays and references in ML programs.
However, it is unclear whether any of those approaches can leverage Rust's type system---especially the borrow mechanisms---during the analysis.


On the other hand, some recent studies build advanced type systems upon Rust's 
type system.
Flux~\cite{Flux} is a liquid-style refinement type system for Rust, which aims to
automatically verify the functional correctness of safe Rust programs.
Flux allows programmers to annotate types (including reference types) with decidable
logical predicates and then exploits Rust's borrow mechanisms to separate functional
verification from low-level memory reasoning.
RefinedRust~\cite{RefinedRust} is another refinement type system for Rust, which aims 
to carry out foundational reasoning and produce Coq proofs for the functional
correctness of both safe and unsafe Rust programs in a semi-automated way.
However, those approaches are focused on functional correctness rather than resource
consumption of Rust programs.


In this paper, we propose \rarust{}, a new AARA approach that leverages Rust's 
borrow mechanisms and a prototype implementation that automatically infers linear 
resource bounds for safe Rust programs.
Rust's type system enforces \emph{ownership} principles around references:
one can \emph{borrow} a reference with ownership and later return it back,
and the type system checks the soundness of such borrows via a notion of \emph{lifetimes}.
Meanwhile, AARA enforces principles of \emph{amortized analysis}~\cite{JADM:Tarjan85}
around data structures: type annotations specify \emph{potential functions},
such that a value can \emph{own} some potentials to pay for resource consumption
and one can \emph{transfer} potentials among values,
and the type system checks the soundness of the potential-based reasoning.
Intuitively, our approach integrates the ownership principles of memory locations
and potentials.
Below we sketch three main technical challenges to obtaining such an integration and our solutions.



Firstly, Rust's type system is rather complex already, and its borrow mechanisms are being
updated actively, so we desire the resource analysis to be as orthogonal to Rust's own
borrow checking as possible.
Rust now features shared borrows, mutable borrows, reborrows,
and non-lexical lifetimes that allow complex uses of borrows.
Modifying Rust's type checker (especially the borrow checker) would be arduous and 
thus unacceptable.
To tackle this challenge, we propose a lightweight design that performs resource analysis
on \emph{well-borrowed} and \emph{well-typed} Rust programs.
\rarust{} relies on Rust's borrow checker to ensure the analyzed program to be well-borrowed
and then exploits Rust's borrow mechanism to associate data structures and references with
potentials.
We adapt the recently proposed Low-Level Borrow Calculus (LLBC)~\cite{Aeneas},
which is based on Rust's MIR, formulate a resource-aware version called Resource-Aware Borrow Calculus (RABC), and devise an AARA type system on RABC.


Secondly, Rust's borrow mechanisms, especially mutable borrows, pose a non-trivial problem for AARA in tracking potentials along borrowing.
For example, one can mutably borrow a value $v$ from a location $\ell$ with $p(v)$ units 
of potentials, update the original value to $v'$, consume or store some potentials, and end the 
borrow to return $v'$ back to the location $\ell$ with $p'(v')$ units of potentials.
Note that the pre- and post-potential functions $p$ and $p'$ can be different; however,
in a typical AARA type system, types are \emph{invariant}, i.e., once a location $\ell$
has a resource-annotated type that specifies a potential function, the type would not change.
To tackle this challenge, we propose \emph{prophecy potentials}, which are inspired from
prophecy variables~\cite{ProphecyInSepLogic} and RefinedRust's borrow names~\cite{RefinedRust}.
Intuitively, prophecy variables provide a compositional reasoning mechanism for \emph{future}
states, as they stand for the value, returned to a location after a mutable borrow ends.
When a mutable borrow of a location $\ell$ with the potential function $p$ starts, \rarust{} generates a prophecy potential function $q$ for $\ell$ and tracks both $p$ and $q$ in
the mutable-borrow reference type.
When a mutable borrow with potential $p'$ and prophecy $q$ ends, \rarust{} enforces $p'$ to denote at least the same amount of potentials as $q$ denotes.


Thirdly, similarly to other static analyses of memory-manipulating programs, we desire
the resource analysis to deal with \emph{aliasing} soundly and precisely.
Rust's borrow checker already ensures some good properties about aliasing, e.g.,
aliasing and mutation cannot happen simultaneously, so it seems sound to always perform
precise \emph{strong updates} for mutation.
However, the type system may face a situation where a mutable reference points to
multiple locations, e.g., a conditional expression whose two branches borrow from two
difference locations.
In this case, the type system should not perform strong updates to the locations,
especially when the program stores some potentials via the mutable reference: it is unsound to increment both locations with the potentials.
To tackle this challenge, we propose a \emph{lattice} of resource-annotated types (including reference types) based on a subtyping relation.
\rarust{} then uses the lattice to compute the greatest lower bound of types for
a location when different control-flow paths join at a program location.



To summarize, we propose \rarust{}, which consists of a resource-aware calculus RABC for well-borrowed Rust programs and an AARA type system with novel prophecy potentials.
We formally prove the soundness of the type system with respect to RABC
and
implement a prototype of \textsc{RaRust} for inferring linear bounds for safe Rust programs.
We collect a suite of Rust programs to evaluate the effectiveness of the prototype, including non-trivial manipulations of mutable lists and trees.
Evaluation results show that the prototype successfully infers symbolic resource bounds. 

\paragraph{Contributions}
The paper makes the following three contributions:
\begin{itemize}
    \item {We propose \rarust{}, which aims at automatic amortized resource analysis for Rust programs. \rarust{} features a lightweight design that works on well-borrowed programs, novel prophecy potentials to handle mutable borrows, and a type lattice to deal with aliasing.}
    \item {We present a formulation of \rarust{}, including Resource-Aware Borrow Calculus (RABC) and an AARA type system. We also formally prove the soundness of the type system.}
    \item {We implement a prototype of \rarust{}, evaluate it on a suite of non-trivial Rust programs, and show that it can infer symbolic resource bounds effectively.}
\end{itemize}


\paragraph{Limitations}
Despite the contributions,
it is worth noting that \rarust{} currently still has some limitations:
(i) \rarust{} does not support unsafe code, which is a key feature of Rust,
(ii) \rarust{} does not support generic types, higher-order functions (i.e., closures), or trait objects,
(iii) \rarust{} does not support non-linear resource bounds.
We will discuss those limitations and sketch pathways for overcoming
them in \cref{sec:discussion}.


\paragraph{Organization}
The paper is organized as follows.
\cref{sec:background} briefly reviews Rust's borrow mechanisms and automatic amortized resource analysis (AARA).
\cref{sec:overview} uses a few running examples to explain the design and key ideas of \rarust{}.
\cref{sec:calculus} formulates Resource-Aware Borrow Calculus (RABC) and the resource-aware dynamic semantics.
\cref{sec:inference} presents an AARA type system for Rust's borrow mechanisms and a type-inference algorithm.
\cref{sec:soundness} sketches the soundness proof of the AARA type system with respect to RABC.
\cref{sec:impl} describes the prototype implementation and experimental evaluation.
\cref{sec:discussion,sec:related} discuss limitations and related work, respectively. 
\cref{sec:conclusion} concludes.

\section{Background}
\label{sec:background}

In this section, we present a brief review of Rust's borrow mechanisms in \cref{sec:background:rust} and the key concepts for automatic amortized resource analysis (AARA) in \cref{sec:overview:AARA}.

\begin{figure}[t]
\footnotesize
\centering
\hrule
\begin{subfigure}[b]{0.49\textwidth}
\begin{lstlisting}[language=Rust, style=colouredRust]
let mut n = 0;
let p1 = &n; // OK, shared borrow from n
// --- p1 lifetime begins ---
let p2 = &n; // OK, shared borrow from n
// --- p2 lifetime begins ---
// --- p2 lifetime ends ---
// --- p1 lifetime ends ---
\end{lstlisting}
\caption{Shared \& Shared}\label{fig:shared-shared}
\end{subfigure}
\begin{subfigure}[b]{0.49\textwidth}
\begin{lstlisting}[language=Rust, style=colouredRust]
let mut n = 0;
let p1 = &n; // OK, shared borrow from n
// --- p1 lifetime begins ---
*p1 = 42; // ERROR, p1 is not mutable
/* drop(p1) */
// --- p1 lifetime ends ---
let p2 = &mut n; // OK, mutable borrow from n
*p2 = 42; // OK, p2 is mutable
\end{lstlisting}
\caption{Shared \& Mutable}\label{fig:shared-mutable}
\end{subfigure}
\hrule
\begin{subfigure}[b]{0.49\textwidth}
\begin{lstlisting}[language=Rust, style=colouredRust]
let mut n = 0;
let p1 = &mut n; // OK, mutable borrow from n
// --- p1 lifetime begins ---
let p2 = &mut n; // ERROR
// n is already mutable-borrowed
*p1 = 42;
// --- p1 lifetime ends ---    
\end{lstlisting}
\caption{Mutable \& Mutable}\label{fig:mutable-mutable}
\end{subfigure}
\begin{subfigure}[b]{0.49\textwidth}
\begin{lstlisting}[language=Rust, style=colouredRust]
let mut n = 0;
let p1 = &mut n; // OK, mutable borrow from n
// --- p1 lifetime begins ---
let p2 = &n; // ERROR
// n is already mutable-borrowed
let m = do_something_with(p2);
*p1 = 42;
// --- p1 lifetime ends ---
\end{lstlisting}
\caption{Mutable \& Shared}\label{fig:mutable-shared}
\end{subfigure}
\hrule
\begin{subfigure}[b]{0.53\textwidth}
\begin{lstlisting}[language=Rust, style=colouredRust]
fn consume(p: &mut i32){ *p = 1; }
fn reborrow() {
  let mut n = 0;
  let p1 = &mut n; // OK, mutable borrow from n
  // --- p1 lifetime begins ---
  let p2 = &mut *p1; // OK, mutable reborrow from p1
  // --- p2 lifetime begins , p1 NOT accessible ---
  consume(p2);
  // --- p2 lifetime ends   , p1 accessible ---
  *p1 = 42;
  // --- p1 lifetime ends ---
}
\end{lstlisting}
\caption{Reborrow}\label{fig:reborrow}
\end{subfigure}
\begin{subfigure}[b]{0.46\textwidth}
\begin{lstlisting}[language=Rust, style=colouredRust]
enum List { Nil, Cons(i32, Box<List>) }
fn iter(l: &List) {
  match l {
    Nil => {
      tick(1); // consume 1 unit of resource
    }, 
    Cons(hd, tl) => { 
      tick(2); // consume 2 units of resource
      iter(&**tl); 
    }, 
  }
}
\end{lstlisting}
\caption{List Iteration}\label{fig:list-iteration}
\end{subfigure}
\hrule
\caption{Examples to Demonstrate Rust's Borrow Mechanisms and Automatic Amortized Resource Analysis}
\label{fig:borrows}
\end{figure}

\subsection{Rust's Borrow Mechanisms}
\label{sec:background:rust}

Rust's type system enforces ownership principles, e.g., declaring a variable \verb|x|
of some type means that \verb|x| is the \emph{exclusive} owner of the memory locations indicated by the type.
Similarly to linear or affine typing~\cite{kn:Walker02}, such exclusive ownership requires a \emph{move} semantics, which could soon be too restrictive in the common call-by-value paradigm.
Fortunately, as C has pointers and C++ has references, Rust introduces \emph{borrow} mechanisms around reference types to allow temporary accesses to the memory locations owned by others.
Rust employs a complex borrow checker to ensure that the borrows in a program work in a memory-safe way.


Rust provides two main kinds of borrows: one is called \textbf{shared} or immutable borrows, the other exclusive or \textbf{mutable} borrows.
We demonstrate those borrow mechanisms using the example programs listed in \cref{fig:borrows}.
Bindings \verb|let| and \verb|let mut| introduce immutable and mutable variables, respectively, and the variables own the memory that stores the data.
\cref{fig:shared-shared} shows that one can use \verb|&n| to create a shared borrow from \verb|n|, and there can be multiple shared borrow from the same variable simultaneously.
\cref{fig:shared-mutable} shows that one can use \verb|&mut n| to create a mutable borrow from \verb|n|, and one cannot mutate the value of \verb|n| via shared borrows.
\cref{fig:mutable-mutable} shows that there can be at most one mutable borrow from the same variable at the same time.
\cref{fig:mutable-shared} shows that if there exists a mutable borrow, there cannot be any shared borrows from the same variable.
Rust's borrow checker introduces a notion of \emph{lifetimes} to track and check borrows.
In \cref{fig:borrows}, we annotate comments to indicate the lifetimes of certain borrows, and Rust's borrow checker performs lifetime inference to automatically derive the lifetimes.\footnote{The demonstration here is simplified for brevity and it is not consistent with the actual Rust's borrow checker.}
Note that in \cref{fig:shared-mutable}, Rust's borrow checker would implicitly insert \verb|drop(p1)| to end the lifetime of the shared reference \verb|p1|, enabling the next \verb|let| binding to create a mutable borrow \verb|p2| from the same variable.

The exclusiveness of mutable borrows could also be too strict: it is not uncommon to lend a mutable reference to a function call and then continue to use the mutable reference after the function returns.
Rust features a mechanism named \textbf{reborrowing}, which is done by creating borrows from dereferences, e.g., \verb|&mut *p1| as shown in \cref{fig:reborrow}.
In this example, \verb|p2| creates a mutable reborrow from \verb|p1|, and \verb|p1| becomes temporarily inaccessible during \verb|p2|'s lifetime because
\verb|p2| is mutable.
This does not break the exclusiveness of mutable borrows, because there is still at most one mutable reference can mutate the data.
Note that creating shared reborrows from dereferences is permitted, whereas
creating mutable reborrows from dereferences of shared references is not.



\subsection{Automatic Amortized Resource Analysis}
\label{sec:overview:AARA}

Automatic amortized resource analysis (AARA) is initially proposed by \citet{AARA-Linear} as a type system that derives linear bounds on the heap-space usage of functional programs.
In the past two decades, AARA has been extended with many features from different perspectives, and becomes a systematic methodology for resource analysis~\cite{JMSCS:HJ22}.
Although most of AARA type systems are presented with pure functional languages, e.g., Resource-aware ML~\cite{RaML},
in this section, we review the key concepts for AARA with a Rust program that uses  shared borrows.
As demonstrated in \cref{sec:background:rust}, Rust permits multiple simultaneous shared borrows from the same variable but disallows mutation via shared borrows; such behavior is similar to pure functional programming.


\cref{fig:list-iteration} implements list iteration via shared references.
A value of the \verb|Box| type owns a heap-allocated piece of data, adhering to Rust's ownership principles.
The \verb|Box| type provides a convenient way to represent heap-allocated data structures, such as lists and trees.
We use the statement \verb|tick(q)|  to indicate \verb|q| units of resource consumption.\footnote{In our current formalization and implementation, we use explicit $\kwd{tick}(q)$ as the cost model. It is also possible to adopt a parametric cost model that maps an expression or a statement to its resource consumption~\cite{RaML}.}
Because the \verb|iter| function consumes two units of resource per \verb|Cons| node and one unit for the \verb|Nil| node, the total resource consumption of applying \verb|iter| to a list of length $n$ is $2n+1$.
To automatically infer such a symbolic resource bound, AARA relies on the \emph{potential method} of amortized complexity analysis~\cite{JADM:Tarjan85},
i.e., AARA relies on the derivation of potential functions that map program states to non-negative numbers, in the sense that the potential of a state is sufficient to pay for the current step's resource consumption and the next state's potential.
AARA introduces resource-annotated types to encode certain kinds of potential functions in type signatures, and then reduces type inference to constraint solving (e.g., linear programming) to automate the bound inference.


To demonstrate the inference process, let us extend the list type
$\kwd{list}$ to a resource-annotated type $\kwd{list}(\alpha)$,
where $\alpha$ is a non-negative number and the type encodes a
linear potential function that assigns $\alpha$ units of potential to
each list element, i.e., the potential for a list of length $n$ is $\alpha \cdot n$.
Readers might realize the inconsistency between the abstract syntax $\kwd{list}$ and the concrete syntax:
\begin{lstlisting}[language=Rust, style=colouredRust, xleftmargin=\leftmargini]
enum List { Nil, Cons(i32, Box<List>) }
\end{lstlisting}
Here, we only annotate $\kwd{list}$ with one symbol $\alpha$ instead of two corresponding to the two constructors \verb|Nil| and \verb|Cons| to simplify the presentation.
Our implementation uses the multiple-symbols version, annotating the list type as follows:
\begin{lstlisting}[language=Rust, style=colouredRust, xleftmargin=\leftmargini, mathescape]
enum List { Nil : $\alpha_{\texttt{Nil}}$, Cons(i32, Box<List>) : $\alpha_{\texttt{Cons}}$ }
\end{lstlisting}
AARA systems can generally analyze different user-defined data types by annotating each constructor with one symbol denoting its potential.
To review AARA, we use the simplified version $\text{list}(\alpha)$, where $\alpha$ is essentially the annotation $\alpha_\texttt{Cons}$ for the \verb|Cons| constructor.

To initiate the analysis of \verb|iter|, we create a signature with two numeric unknowns $\alpha$ and $\delta$:
\[
\verb|iter| : \kwd{fn}(\verb|l|:\&\kwd{list}(\alpha)) \to () | \delta ,
\]
where $\delta$ denotes additional potential needed by \verb|iter| that is not associated with the list's potential function.
We then analyze the function body of \verb|iter| to collects constraints on $\alpha$ and $\delta$.
The pattern match on \verb|l| peeks the list structure; thus, according
to the definition of $\kwd{list}(\alpha)$, the \verb|Cons| branch would obtain $\alpha$ units of potential.
We proceed to analyze the two branches.
\begin{itemize}
    \item In the \verb|Nil| branch, we have $\delta$ units of potential and the statement \verb|tick(1)| consumes one unit of resource.
    According to the potential method, we collect one constraint $\delta \ge 1$.

    \item In the \verb|Cons| branch, we have $\delta + \alpha$ units of potential, and the variables \verb|hd| and \verb|tl| have type $\&\kwd{i32}$ and $\&\kwd{box}~\kwd{list}(\alpha)$, respectively.
    The statement \verb|tick(2)| consumes two units of resource, followed by a recursive call with \verb|&**tl|.
    We now use the signature of \verb|iter| to analyze this recursive call; that is, it requires \verb|&**tl| to have type $\&\kwd{list}(\alpha)$ as well as $\delta$ units of additional potential.
    The first requirement is fulfilled immediately as \verb|tl| has type $\&\kwd{box}~\kwd{list}(\alpha)$, and the second requirement can be satisfied if $\delta + \alpha \ge 2 + \delta$.
\end{itemize}
Noting that the constraints $\delta \ge 1, \delta + \alpha \ge 2 + \delta$ are linear constraints, we can employ a linear programming (LP) solver to automatically find instances of $\alpha$ and $\delta$.
For example, with a proper objective, an LP solver would find the assignments $\alpha = 2, \delta=1$, which yield the following signature:
\[
\verb|iter| : \kwd{fn}(\verb|l|:\&\kwd{list}(2)) \to () | 1 ,
\]
which coincides with the bound $2n+1$ where $n$ is the length of the list \verb|l|.

\section{Overview}
\label{sec:overview}

In this section, we use a few running examples to explain
the core ideas behind \rarust{} that integrates AARA with Rust's
borrow mechanisms.
\cref{sec:overview:Shared} shows how \rarust{} deals with \textbf{shared} borrows.
\cref{sec:overview:Mutable} shows how \rarust{} deals with \textbf{mutable} borrows.
\cref{sec:overview:Lattice} shows how \rarust{} deals with the aliasing problem.
Recall that we propose a lightweight design for \rarust{}, which assumes the programs already pass Rust's borrow checking so that \rarust{} works directly on \emph{well-borrowed} and \emph{well-typed} Rust programs.
Concretely, \rarust{} assumes that all borrows in the analyzed program have known lifetimes (the span they live) and satisfy the following properties:
\begin{itemize}
    \item multiple shared but no mutable borrows from the same piece of data are live at the same time; or
    \item no shared but at most one mutable borrow from the same piece of data are live at the same time.
\end{itemize}
We will show how \rarust{} exploits those properties to carry out AARA for Rust programs.



\begin{figure}[t]
\centering
\footnotesize
\hrule
\begin{subfigure}[b]{0.52\textwidth}
\begin{lstlisting}[language=Rust, style=colouredRust]
fn iter_twice(l: &List) {
  // l : &list(4)
  iter(&*l); // share 4 as 2 + 2, &*l : &list(2)
  // l : &list(2)
  iter(&*l); // share 2 as 2 + 0, &*l : &list(2)
  // l : &list(0)
}
\end{lstlisting}
\caption{Shared Reborrowing}
\label{fig:ex-sharing}
\end{subfigure}
\begin{subfigure}[b]{0.46\textwidth}
\begin{lstlisting}[language=Rust, style=colouredRust]
fn update(l: &mut List) {
  iter(&*l);
  // l : &mut list(0)
  *l = Cons(3, Box::new(Nil));
  // l : &mut list(4)
  iter(&*l); iter(&*l);
}
\end{lstlisting}
\caption{Mutating A Mutable Borrow}
\label{fig:ex-mut-borrow}
\end{subfigure}
\hrule
\begin{subfigure}[b]{0.48\textwidth}
\begin{lstlisting}[language=Rust, style=colouredRust]
fn prophecy() {
  let mut l = Cons(3, Box::new(Nil));
  // l : list(p)
  let x = &mut l;
  // l : list(q),  x : &mut(list(p ), list(q))
  update(x);    // x : &mut(list(p'), list(q))
  /* drop(x) */ // p' >= q
}
\end{lstlisting}
\caption{Creating \& Dropping A Mutable Borrow}
\label{fig:ex-prophecy}
\end{subfigure}
\begin{subfigure}[b]{0.51\textwidth}
\begin{lstlisting}[language=Rust, style=colouredRust]
fn weak(b: bool, l1: &mut List, l2: &mut List) {
  let l = if b { 
    &mut *l1 // : &mut(list(p1), list(q1))
  } else { 
    &mut *l2 // : &mut(list(p2), list(q2))
  }; // : &mut(list(min(p1,p2)), list(max(q1,q2)))
  update(l);
}
\end{lstlisting}
\caption{Mutable Reborrowing \& Aliasing}
\label{fig:ex-weak}
\end{subfigure}
\hrule
\caption{Examples to Demonstrate How \rarust{} Works}
\label{fig:running-examples}
\end{figure}

\subsection{Dealing with Shared Borrows via Shared Potentials}
\label{sec:overview:Shared}

In \cref{sec:overview:AARA}, we demonstrate the key concepts of AARA via
the Rust program shown in \cref{fig:list-iteration}, which already features shared borrows.
However, this is not the end of the story: multiple shared borrows from the same memory location can exist simultaneously; for example, the function
\verb|iter_twice| shown in \cref{fig:ex-sharing} uses the reborrowing mechanism
to create two more shared borrows by \verb|&*l|.
If we still follow the methodology presented in \cref{sec:overview:AARA},
suppose that the function parameter \verb|l| has type $\&\kwd{list}(\alpha)$,
it would be unsound to type the two shared reborrows \verb|&*l| as $\&\kwd{list}(\alpha)$, because it would double the potentials stored in \verb|l|.



Fortunately, prior research on AARA has proposed a notion of \emph{sharing} to allow multiple uses of a variable in a linear or affine type system~\cite{AARA-Poly-Multivar,AARA-Poly}. 
This is because AARA for functional programs needs shared potentials to pass value by reference.
The idea is to replace a variable $x$ of resource-annotated type $T$
with two fresh variables $x_1,x_2$ of types $T_1,T_2$, such that the potential function denoted by $T$ equals the sum of the potential functions denoted by $T_1,T_2$.
In RaRust, we reuse the sharing mechanism to handle shared borrows and prove it is sound with respect to the safe Rust semantics.
In our setting, this indicates that we can replace a typing context
$x : \kwd{list}(\alpha)$ with $x_1 : \kwd{list}(\alpha_1), x_2:\kwd{list}(\alpha_2)$ such that $\alpha = \alpha_1 + \alpha_2$.

In \rarust{}, we adopt a more imperative design inspired by \emph{remainder contexts}~\cite{kn:Walker02,ICFP:KH21}.
We associate every program point with a typing context, and when a statement performs a shared (re)borrow, we split the potentials into two parts by splitting the resource-annotated type into two types as shown above: one becomes the type
of the shared (re)borrow, and the other is put back into the remainder context, i.e., the typing context after the statement.
For example, in \cref{fig:ex-sharing}, suppose that the function parameter \verb|l| has type $\&\kwd{list}(4)$, the first function call to \verb|iter|
performs a shared reborrow and we split the type $\&\kwd{list}(4)$ to $\&\kwd{list}(2)$ (for the function call) and $\&\kwd{list}(2)$ (for the remainder context).
The second function call also performs a shared reborrow, but this time the typing context indicates that \verb|l| has type $\&\kwd{list}(2)$, so we split it to $\&\kwd{list}(2)$ (for the function call) and $\&\kwd{list}(0)$.
Observing that the function \verb|iter| requires one unit of additional potentials, we derive the following signature for \verb|iter_twice|:
\[
\verb|iter_twice| : \kwd{fn}(\verb|l|:\&\kwd{list}(4)) \to () | 2 .
\]


It is worth noting that Rust's shared borrows are \emph{immutable}.
However, \rarust{}'s analysis \emph{mutates} the resource-annotated types of shared borrows because shared borrows could consume resources when being accessed. 
This is safe unless the value of the borrow is mutated.
For example, a shared borrow \verb|l| points to a list of length 10 that carries 2 units of potentials per element; thus, the total potentials are 20 units.
We then create another shared borrow \verb|&*l| and split the potentials to let \verb|&*l| carry one unit of potentials per element.
Now suppose we mutate the value via the borrow \verb|l| to increment the list length by one.
The type of \verb|&*l| still indicates one unit of potentials per element, thus indicating 11 units of total potentials, but it only has 10 units.
To tackle the problem, \rarust{} exploits Rust's borrow mechanisms to render the reasoning sound:
mutable borrows and shared borrows from the same memory location cannot exist simultaneously.
Thus, if a Rust program mutates the list via a mutable reference \verb|l|, then the previous shared borrow \verb|&*l| must have ended its lifetime. 



\subsection{Dealing with Mutable Borrows via Prophecy Potentials}
\label{sec:overview:Mutable}

It might seem straightforward to support mutable borrows in the approach sketched in \cref{sec:overview:Shared}.
For example, \cref{fig:ex-mut-borrow} implements a function \verb|update| that manipulates a mutable reference \verb|l|.
Suppose that the function parameter \verb|l| has type $\&\kwd{mut}~\kwd{list}(2)$.
For the first function call to \verb|iter| with a shared reborrow \verb|&*l|, we split the type to $\&\kwd{list}(2)$ and $\&\kwd{mut}~\kwd{list}(0)$.
The next assignment statement mutates the list stored in the location referenced by \verb|l|, so we mutate its type in the typing context accordingly.
To obtain enough potential for the remaining two function calls to \verb|iter|, we set the type of \verb|l| to $\&\kwd{mut}~\kwd{list}(4)$.
Because the new list \verb|*l| is a singleton list, the mutation itself consumes 4 units of potentials.
Similarly to the reasoning in \cref{sec:overview:Shared}, the potential is sufficient to perform the remaining two function calls, and the final remainder context is $\verb|l|: \&\kwd{mut}~\kwd{list}(0)$.
Noting that three calls to \verb|iter| need three units of additional potentials, we derive the following signature for \verb|update|:
\[
\verb|update|: \kwd{fn}(\verb|l|:\&\kwd{mut}~\kwd{list}(2)) \to () | 7 .
\]

A tricky issue arises when one considers creating and dropping mutable borrows.
\cref{fig:ex-prophecy} shows an example where the program creates a mutable borrow \verb|x| from a mutable list \verb|l|.
Note that it is no longer sound to split the potentials of \verb|l| into two parts and store one part in \verb|x|: the reason is stated already at the end of \cref{sec:overview:Shared}; that is, the program can later mutate the list \verb|l| through the mutable reference \verb|x|, making the remainder type of \verb|l| unsound.
Fortunately, Rust's borrow mechanisms ensure a good property that
at most, one mutable borrow from the same memory location can be live simultaneously, so in principle, it would be possible to track the change in the type of the mutable borrow \verb|x| and pass the change back to \verb|l| when \verb|x| gets dropped.
It is worth noting that our type system is aware of when a borrow $x$ gets dropped via an explicit statement $\kwd{drop}~ x$, which is generated according to lifetime constraints given by the Rust compiler.





One idea might emerge immediately that the resource-annotated type of a mutable borrow keeps the location where it borrows from, as $\verb|l|:\&\kwd{mut}(\kwd{list}(\alpha), \verb|p|)$ with \verb|p| be the location such that \lstinline|l = &mut p|.
However, this design essentially embeds a pointer analysis in the type system, and one would soon find out that every mutable reference type needs to record a \emph{set} of possible locations.
\cref{fig:ex-weak} exemplifies this case and we will revisit this example in \cref{sec:overview:Lattice}.
Because those locations are usually local variables, it becomes unclear how to carry out inter-procedural analysis in a compositional way, and we certainly do not want function signatures to reveal local variables.


Such an issue is not uncommon in the studies of advanced type systems or verification frameworks for Rust.
\citet{ProphecyInSepLogic} adapt \emph{prophecy variables}---which were originally proposed by \citet{LICS:AL88} to talk about future's program states during reasoning---to integrate separation logic with prophecies.
RustHorn~\cite{RustHorn} and RefinedRust~\cite{RefinedRust} also use prophecy variables in their verification frameworks.
The high-level idea of using prophecy variables to analyze Rust's mutable borrows is to record additional information which corresponds to the final value of a reference, i.e., the value when the reference gets dropped. 
However, prophecy variables usually record the final values, which are too heavy for AARA.

In \rarust{}, we propose \emph{prophecy potentials}, a novel adaption of prophecy variables to the AARA methodology to reason about the future's potential functions.
\cref{fig:ex-prophecy} shows how prophecy potentials work with mutable borrows.
We now represent the type of mutable reference as $\&\kwd{mut}(\tau_1,\tau_2)$, where $\tau_1$ denotes the current potential function and $\tau_2$ denotes the prophecy potential function, i.e., the expected potential function when the lifetime of the reference ends.
In \cref{fig:ex-prophecy}, suppose that the initial type of \verb|l| is $\kwd{list}(p)$ with some $p \ge 0$.
To create a mutable borrow \verb|x| from \verb|l|, we generate a prophecy type $\kwd{list}(q)$ with some $q \ge 0$, which indicates the final resource-annotated type of \verb|x| when \verb|x| gets dropped.
The mutable reference type of \verb|x| is initialized to $\&\kwd{mut}(\kwd{list}(p), \kwd{list}(q))$.
After the call to the function \verb|update|, the type of \verb|x| becomes $\&\kwd{mut}(\kwd{list}(p'), \kwd{list}(q))$.
Note that the prophecy type should remain unchanged.
When the mutable reference \verb|x| drops, \rarust{} emits a constraint that the potentials indicated by $\kwd{list}(p')$ are no less than the potentials indicated by the prophecy type $\kwd{list}(q)$, i.e., $p' \ge q$.
If \verb|l| would later be used again, we can start from the type $\kwd{list}(q)$.
In this way,
prophecy potentials enable \emph{compositional} reasoning about mutable borrows.

\subsection{Dealing with Aliasing via A Lattice of Resource-Annotated Types}
\label{sec:overview:Lattice}

A type system sometimes cannot precisely determine where a mutable reference is borrowed from.
For example, \cref{fig:ex-weak} uses a conditional expression to assign a mutable reference \verb|l| to a mutable reborrow from either \verb|l1| or \verb|l2|, depending on the runtime value of the Boolean-valued variable \verb|b|.
As the example shows, although Rust's borrow mechanisms enjoy some good properties that aid our design of \rarust{}, we still face the problem of \emph{aliasing}, as other static analyses of heap-manipulating programs would also face.

In our work, we can still exploit Rust's borrow mechanisms, which ensure that aliasing and mutation cannot happen at the same time.
Therefore, we only need to consider \emph{control-flow aliasing}; that is,
when the control-flow paths merge at a program point, we need to \emph{merge} the types---including the mutable reference types---of a variable from different paths.
The merging here needs to be conservative, similarly to the \emph{weak updates} usually seen in pointer analyses.
In \rarust{}, we formulate a subtyping relation among resource-annotated types to formalize the notion of ``conservative'' and then construct a \emph{lattice} of resource-annotated types based on subtyping to carry out merging.
For example, in \cref{fig:ex-weak}, suppose that the mutable reborrows \verb|&mut *l1| has type $\&\kwd{mut}(\kwd{list}(p_1),\kwd{list}(q_1))$ and \verb|&mut *l2| has type $\&\kwd{mut}(\kwd{list}(p_2),\kwd{list}(q_2))$.
Recall that $\kwd{list}(q_1)$ and $\kwd{list}(q_2)$ are prophecy types.
To obtain the type of the mutable reference \verb|l|, we need to merge the two types above.
Thinking about the merging conservatively, one can derive that
\verb|l| can hold potentials no more than the potentials indicated by
$\kwd{list}(p_1)$ and $\kwd{list}(p_2)$, so \verb|l|'s current potential type is
at most $\kwd{list}(\min(p_1,p_2))$.
Meanwhile, to ensure that the prophecies are sound no matter which branch is executed, \verb|l|'s prophecy potential type should be at least $\kwd{list}(\max(q_1,q_2))$.
In addition, because the type system cannot know which of \verb|l1| and \verb|l2| would be mutated later or which of them would remain unchanged, \rarust{} enforces that $p_1 \ge q_1$ and $p_2 \ge q_2$.

As illustrated above, \rarust{}'s current mechanism of handling aliasing compromises the precision of the resource analysis, mostly due to weak updates.
On the one hand, such precision loss is unavoidable---at least for \cref{fig:ex-weak}---due to insufficient information about runtime values during type checking.
On the other hand, recent work such as Flux~\cite{Flux} introduces \emph{strong references} to perform strong updates, and it is interesting future research to adapt them in \rarust{}.


\section{Resource Aware Borrow Calculus}
\label{sec:calculus}

This section introduces Resource-Aware Borrow Calculus (RABC) and resource-aware dynamic semantics.
RABC is a resource-aware variant of Low-Level Borrow Calculus (LLBC)~\cite{Aeneas}, which is based on Rust's MIR but keeps high-level information such as structured control flow and a structured memory model.
RABC includes essential features such as mutation, borrow mechanisms, integer lists with explicit boxing, and recursive top-level functions.\footnote{We include only integer lists as heap-allocated data structures and exclude loops in RABC for ease of presentation. Our implementation supports user-defined inductive data types using structs and enums, as well as \texttt{while true} loops with break and continue.}
RABC also includes $\kwd{tick}(\cdot)$ statements to annotate resource consumption.
\cref{sec:syntax} presents the syntax and discusses properties of a well-borrowed RABC program, guaranteed by Rust's borrow checker.
\cref{sec:semantics} formalizes the resource-aware dynamic semantics of RABC, which captures the resource consumption during the execution of an RABC program.

\subsection{Syntax}
\label{sec:syntax}

\cref{fig:syntax} summarizes the syntax of RABC. For the convenience of formalization, we distinguish between expressions and statements.
We then describe each syntactical construction separately. 

\begin{figure}[t]
\small
    \begin{align*}
    \textbf{Type}~ t &::= \\
        \tag{atom} &|~ \kwd{i32} ~|~ \kwd{bool} \\
        \tag{list} &|~ \kwd{list} ~|~ \kwd{box}~\kwd{list}\\
        \tag{borrow} &|~ \&^\kwd{s}~t ~|~ \&^\kwd{m}~t \\
    \textbf{Place}~ p &::= \\
        \tag{variable} &|~ x \\
        \tag{dereference} &|~ * p \\
    \textbf{Expression}~ e &::= \\
        &|~ \kwd{n}_\text{i32} ~|~ \kwd{true} ~|~ \kwd{false} ~|~ \kwd{nil} ~|~ \kwd{box}(e) \\
        \tag{integer} &|~ e_1 ~\kwd{op}~ e_2 \\
        \tag{scalar copy} &|~ \kwd{copy}~ p \\
        \tag{borrow} &|~ \&^\kwd{s}~ p ~|~ \&^\kwd{m}~ p \\
        \tag{move ownership} &|~ \kwd{move}~ p \\
    \textbf{Statement}~ s &::= \\
        \tag{resource cost} &|~ \kwd{tick}(\delta) \\
        \tag{return} &|~ \kwd{return} \\
        \tag{sequence} &|~ s_1; s_2 \\
        \tag{drop} &|~ \kwd{drop}~ p\\
        \tag{if bool} &|~ \kwd{if}~ p ~\kwd{then}~ s_1 ~\kwd{else}~ s_2 ~\kwd{end}\\
        \tag{match list} &|~ \kwd{match}~ p ~ \{\kwd{nil}\mapsto s_1, \kwd{cons}(x_\text{hd}, x_\text{tl})\mapsto s_2\} \\
        \tag{assignment} &|~ p \from e \\
        \tag{list constructor} &|~ p \from \kwd{cons}(e_1, e_2) \\
        \tag{function call} &|~ p \from f(\vec{e}) \\
    \textbf{Toplevel}~ tl &::= \\
        \tag{sequence} &|~ tl_1 ~ tl_2 \\
        \tag{function} &|~ \kwd{fn}~ f ~(\vec{x}_\text{param}:\vec{t}_\text{param}, \vec{x}_\text{local}:\vec{t}_\text{local}, x_\text{ret}:t_\text{ret}) \{~ s ~\}
    \end{align*}
    \caption{Syntax}
    \label{fig:syntax}
\end{figure}

\textbf{Types} are simple, without resource annotations; we will present the types with annotations in \cref{sec:inference}. Integer $\kwd{i32}$ and Boolean $\kwd{bool}$ are atom types. Lists $\kwd{list}$ and the box type of lists $\kwd{box}~\kwd{list}$ are types for functional lists defined in Rust. The box type is required for a list's tail, which is usually heap-allocated. The reference type is for borrows with different modes: $\&^\kwd{s}~ t$ is for the shared borrow, and $\&^\kwd{m}~ t$ is for the mutable borrow. These borrow modes are notions from Rust explained as follows: mutation is forbidden on shared borrows and only allowed on mutable borrows.

\textbf{Places} are memory locations storing values, including program variables $x, y, \ldots$ and dereferences $* p$ of borrows or boxes stored in $p$. We will soon show their role in the dynamic semantics in \cref{sec:semantics}. 

\textbf{Expressions} represent resource-free evaluation. Integer literals $\kwd{n}_\text{i32}$ and Boolean literals $\kwd{true}$, $\kwd{false}$ are atom values. The $\kwd{nil}$ constructor stands for empty lists. The boxing expression $\kwd{box}(e)$ allocates memory in a heap to store the value of $e$, resembling \lstinline|Box::new(e)| of Rust. Arithmetic expressions $e_1 ~\kwd{op}~ e_2$ operate on integer-valued operands with the binary operator $\kwd{op}$. For an atom value stored in a place $p$, we use $\kwd{copy}~p$ to make a copy of it.
Given a place $p$, we can borrow from it with different modes: $\&^\kwd{s}~ p$ creates a shared borrow, and $\&^\kwd{m}~ p$ creates a mutable borrow.
For a borrow stored in a place $p$, we can use $\kwd{move}~ p$ to move ownership out from the original place $p$. 

\textbf{Statements} represent resource-aware evaluation. The statement $\kwd{tick}(\delta)$ with $\delta \in \ZZ$ is the explicit annotation for consuming $\delta$ units of resource.
Statements $\kwd{return}$ and $s_1; s_2$ usually form a function body such as $s_1; s_2; \ldots, s_n; \kwd{return}$. In RABC, we introduce $\kwd{drop}~ p$ to drop the borrow stored in $p$ explicitly. The conditional and pattern-match statements perform case analysis on Boolean values and lists, respectively. Note that we use place $p$ instead of the expression $e$ to indicate Boolean values and lists because we only need to peak the value instead of copying Boolean values or moving ownership of lists. The assignment has three variants: one for assigning atom values and borrows, one for constructing lists, and another for function applications; the latter two variants do not reside in expressions because they need to be resource-aware as they incur resource consumption.

\textbf{Toplevels} define a sequence of (possibly recursive) top-level functions like $\kwd{fn} ~f_1, \ldots, \kwd{fn}~ f_n$. Each function contains one statement as its body and variables with corresponding type declarations, including function parameters $\vec{x}_\text{param}:\vec{t}_\text{param}$, local variables $\vec{x}_\text{local}:\vec{t}_\text{local}$ used in the function body, and a distinguished variable $x_\text{ret} : t_\text{ret}$ for the returned value. The $\vec{\bullet}$ notation represents vectors.

\subsection{Resource Aware Dynamic Semantics}
\label{sec:semantics}

\begin{figure}[t]
\small
    \begin{align*}
    \tag{undefined} \textbf{Value}~ v &::= \bot \\
    \tag{atoms} &|~ \kwd{n}_\text{i32} ~|~ \kwd{true} ~|~ \kwd{false} \\
    \tag{list} &|~  lv \\
    \tag{box} &|~ \kwd{box}(lv) \\
    \tag{borrow} &|~ \&(p, v)\\
    \textbf{List Value}~ lv &::= \kwd{nil} ~|~ \kwd{cons}(\kwd{n}_\text{i32}, \kwd{box}(lv))
    \end{align*}
    \caption{Value}
    \label{fig:dyn-value}
\end{figure}

\begin{figure}[t]
\small
    \judgement{Store Reading}{$V\vDash p \rightsquigarrow v$}
    \begin{mathpar}
    \inferrule*[Right=\rulename{V-Rd-Var}]
    {V(x)=v}
    {V\vDash x \rightsquigarrow v}
    \and
    \inferrule*[Right=\rulename{V-Rd-Box}]
    {V\vDash p \rightsquigarrow \kwd{box}(v)}
    {V\vDash * p\rightsquigarrow v}
    \and
    \inferrule*[Right=\rulename{V-Rd-Borrow}]
    {V\vDash p \rightsquigarrow \&(\_, v)}
    {V\vDash *p \rightsquigarrow v}
    \end{mathpar}

    \judgement{Store Writing}{$\VWt{V}{p}{v}{V'}$}
    \begin{mathpar}
    \inferrule[V-Wt-Var]
    {\forall y\not=x, V'(y)=V(y) 
    \\ V'(x) = v}
    {\VWt{V}{x}{v}{V'}}
    \and
    \inferrule[V-Wt-Box]
    {V\vDash p\rightsquigarrow \kwd{box}(\_)
    \\ \VWt{V}{p}{\kwd{box}(v)}{V'}}
    {\VWt{V}{* p}{v}{V'}}
    \and
    \inferrule*[Right=\rulename{V-Wt-Borrow}]
    {V\vDash p\rightsquigarrow \&(q, \_)
    \\ \VWt{V}{q}{v}{V'}
    \\ \VWt{V'}{p}{\&(q, v)}{V''}}
    {\VWt{V}{*p}{v}{V''}}
    \end{mathpar}

    \caption{Store Reading and Writing}
    \label{fig:dyn-rw}
\end{figure}

\begin{figure}[t]
\small
\judgement{Expression Evaluation (Selected)}{$V\vDash e \rightsquigarrow v$}
    \begin{mathpar}
    \inferrule*[Right=\rulename{V-Ev-Borrow}]
    {V\vDash p \rightsquigarrow v}
    {V\vDash \&^{\kwd{s}/\kwd{m}} p \rightsquigarrow \&(p, v)}
    \end{mathpar}

\judgement{Statement Execution (Selected)}{$V\vDash e \rightsquigarrow^\delta \Dashv V'$}
    \begin{mathpar}
    \inferrule*[Right=\rulename{V-Ex-Tick}]
    {~}
    {V\vDash \kwd{tick}(\delta)\rightsquigarrow^\delta \Dashv V}
    \and
    \inferrule*[Right=\rulename{V-Ex-Drop}]
    {~}
    {V\vDash \kwd{drop}~p\rightsquigarrow^0\Dashv V}
    \\
    \inferrule*[Right=\rulename{V-Ex-Cons}]
    {V\vDash e_1 \rightsquigarrow v_1
    \\ V\vDash e_2 \rightsquigarrow v_2
    \\ \VWt{V}{p}{\kwd{cons}(v_1, v_2)}{V'} }
    {V\vDash p\from \kwd{cons}(e_1, e_2)\rightsquigarrow^0 \Dashv V'}
    \\
    \inferrule*[Right=\rulename{V-Ex-IfT}]
    {V\vDash p\rightsquigarrow \kwd{true}
    \\ V\vDash s_1\rightsquigarrow^\delta \Dashv V'}
    {V\vDash \kwd{if}~ p ~\kwd{then}~ s_1 ~\kwd{else}~ s_2 ~\kwd{end} \rightsquigarrow^\delta \Dashv V'}
    \and
    \inferrule*[Right=\rulename{V-Ex-IfF}]
    {V\vDash p\rightsquigarrow \kwd{false}
    \\ V\vDash s_2\rightsquigarrow^\delta \Dashv V'}
    {V\vDash \kwd{if}~ p ~\kwd{then}~ s_1 ~\kwd{else}~ s_2 ~\kwd{end} \rightsquigarrow^\delta \Dashv V'}
    \\
    \inferrule*[Right=\rulename{V-Ex-Mat-Nil}]
    {V\vDash p\rightsquigarrow \kwd{nil}
    \\ V\vDash s_1\rightsquigarrow^\delta \Dashv V'}
    {V\vDash \kwd{match}~ p ~ \{\kwd{nil}\mapsto s_1, \kwd{cons}(x_\text{hd}, x_\text{tl})\mapsto s_2\} \rightsquigarrow^\delta \Dashv V'}

    \inferrule*[Right=\rulename{V-Ex-Mat-Cons}]
    {V\vDash p\rightsquigarrow \kwd{cons}(hd, tl)
    \\ \VWt{V}{p}{\bot}{V_1}
    \\ \VWt{V_1}{x_\text{hd}}{hd}{V_2}
    \\ \VWt{V_2}{x_\text{tl}}{tl}{V_\text{b}}
    \\\\ V_\text{b}\vDash s_2\rightsquigarrow^\delta \Dashv V'_\text{b}
    \\ V'_\text{b}\vDash x_\text{hd}\rightsquigarrow hd'
    \\ V'_\text{b}\vDash x_\text{tl}\rightsquigarrow tl'
    \\\\ \VWt{V'_\text{b}}{x_\text{hd}}{\bot}{V'_1}
    \\ \VWt{V'_1}{x_\text{tl}}{\bot}{V'_2}
    \\ \VWt{V'_2}{p}{\kwd{cons}(hd', tl')}{V'} 
    }
    {V\vDash \kwd{match}~ p ~ \{\kwd{nil}\mapsto s_1, \kwd{cons}(x_\text{hd}, x_\text{tl})\mapsto s_2\} \rightsquigarrow^\delta \Dashv V'}
    \\
    
    \inferrule*[Right=\rulename{V-Ex-App}]
    {\kwd{fn}~ f ~(\vec{x}_\text{param}:\vec{t}_\text{param}, \vec{x}_\text{local}:\vec{t}_\text{local}, x_\text{ret}:t_\text{ret}) \{~ s ~\}
    \\ V\vDash \vec{e}\rightsquigarrow \vec{v}
    \\\\ \VWt{V}{\vec{x}_\text{param}}{\vec{v}}{V_1}
    \\ \VWt{V_1}{\vec{x}_\text{local}}{\bot}{V_2}
    \\ \VWt{V_2}{x_\text{ret}}{\bot}{V_\text{b}}
    \\\\ V_\text{b}\vDash s\rightsquigarrow^\delta \Dashv V'_\text{b}
    \\ V'_\text{b}\vDash x_\text{ret} \rightsquigarrow v_\text{ret}
    \\\\ \VWt{V'_\text{b}}{\vec{x}_\text{param}}{\bot}{V'_1}
    \\ \VWt{V'_1}{\vec{x}_\text{local}}{\bot}{V'_2}
    \\ \VWt{V'_2}{x_\text{ret}}{\bot}{V'_3}
    \\ \VWt{V'_3}{p}{v_\text{ret}}{V'}
    } 
    {V\vDash p\from f(\vec{e})\rightsquigarrow^\delta \Dashv V'}
    \end{mathpar}
    \caption{Resource Aware Dynamic Semantics}
    \label{fig:dyn-eval-exec}
\end{figure}

\cref{fig:dyn-value}, \cref{fig:dyn-rw}, and \cref{fig:dyn-eval-exec} define a  resource-aware big-step dynamic semantics for RABC.
\cref{fig:dyn-value} defines values of RABC, including atom values, list values, box values, borrow values, and a distinguished undefined value $\bot$. 
Note that borrow values take the form $\&(p,v)$, denoting a value $v$ borrowed from a place $p$, but do not record the borrow mode ($\&^\kwd{s}$ or $\&^\kwd{m}$).
The design is reasonable because we work on well-borrowed programs; thus, we do not need to track the borrow modes during runtime.

A \emph{store} is a mapping $V : \mathbf{Variable}\to\mathbf{Value}$, where unused variables can be mapped to $\bot$.
\cref{fig:dyn-rw} formalizes reading from and writing to a store.
Judgement $V \vDash p \rightsquigarrow v$ means that under a store $V$, the place $p$ records a value $v$.
Judgement $\VWt{V}{p}{v}{V'}$ means that starting from a store $V$, writing a value $v$ to the place $p$ yields a new store $V'$.
Note that the rule \rulename{V-Wt-Borrow} may not terminate in general for heap-manipulating languages like C.
In our setting, we exploit Rust's borrow mechanisms that ensure that one cannot construct cyclic reference relations using borrows.


\cref{fig:dyn-eval-exec} presents selected evaluation rules for expressions and statements.
Judgement $V\vDash e \rightsquigarrow v$ indicates that under a store $V$, the expression $e$ evaluates to the value $v$. Recall that expressions denote resource-free computation, so we do not record resource information.
%
%
The rule \rulename{V-Ev-Borrow} reflects the design that the runtime does not need to track borrow modes for well-borrowed programs.
Judgement $V \vDash s \rightsquigarrow^\delta \Dashv V'$ means that starting from a store $V$, the statement $s$ executes with $\delta$ units of resource consumption and ends in the store $V'$.
The rule \rulename{V-Ex-Tick} introduces $\delta$ unit of resource consumption; this is the only rule to incur actual resource uses.
The rule \rulename{V-Ex-Drop} does nothing, i.e., it does need to put the value back to the borrowed place, because we immediately update values when writing through borrows, as indicated by the rule \rulename{V-Wt-Borrow} in \cref{fig:dyn-rw}.
Also, because of such immediate updates, it is necessary to make sure that original places and variables should be passed to function applications; the subscript $\textbf{b}$ of the store $V_\text{b}$ stands for \textbf{b}inding in the rule \rulename{V-Ex-App}.
%

\section{Resource Aware Type System and Inference} \label{sec:inference}

In this section, we present the resource-aware type system based on RABC introduced in \cref{sec:calculus} and a type-inference algorithm based on the AARA methodology.
\cref{sec:inference:types} introduces resource-enriched types, which augment the types of RABC with resource annotations.
%
%
\cref{sec:inference:subtyping} formulates a subtyping relation among resource-enriched types and uses the relation to construct a lattice of types sketched in \cref{sec:overview:Lattice}.
%
%
\cref{sec:inference:eval} and \cref{sec:inference:exec} present the resource-aware typing rules for expressions and statements, respectively.
\cref{sec:inference:infer} discusses a type-inference algorithm for the resource-aware type system.

\subsection{Rich Types, Contexts, and Signatures} \label{sec:inference:types}
\begin{figure}[t]
\small
    \begin{align*}
    \tag{undefined} \textbf{RichType}~ \tau &::= \bot \\
    \tag{atom types} &|~ \kwd{i32} ~|~ \kwd{bool} \\
    \tag{list} &|~ \kwd{list}(\alpha)\\
    \tag{box} &|~ \kwd{box}(\kwd{list}(\alpha)) \\
    \tag{shared borrow} &|~ \&^\kwd{s}(\tau) \\
    \tag{mutable borrow} &|~ \&^\kwd{m}(\tau_\text{c}, \tau_\text{p})
    \end{align*}
    \caption{Rich Types}
    \label{fig:rich-type}
\end{figure}
\begin{figure}[t]
\small
    \judgement{Enrich (Selected)}{$\textit{enrich}~ t ~\textit{as}~ \tau$}
    \begin{mathpar}
    \inferrule[Enrich-List]
    {\alpha~\text{fresh}}
    {\textit{enrich}~ \kwd{list} ~\textit{as}~  \kwd{list}(\alpha)}
    \and
    \inferrule[Enrich-Shared]
    {\textit{enrich}~ t ~\textit{as}~ \tau}
    {\textit{enrich}~ \&^\kwd{s}(t) ~\textit{as}~ \&^\kwd{s}(\tau)}
    \and
    \inferrule[Enrich-Mutable]
    {\textit{enrich}~ t ~\textit{as}~ \tau_\text{c}
    \\ \textit{enrich}~ t ~\textit{as}~ \tau_\text{p}
    }
    {\textit{enrich}~ \&^\kwd{m}(t) ~\textit{as}~ \&^\kwd{m}(\tau_\text{c}, \tau_\text{p})}
    \end{mathpar}
    \caption{Enrichment}
    \label{fig:enrich}
\end{figure}

\begin{figure}[t]
\small
    \judgement{Context Reading}{$\Gamma\vdash p \hookrightarrow \tau$}
    \begin{mathpar}
    \inferrule[$\Gamma$-Rd-Var]
    {\Gamma(x)=\tau}
    {\Gamma\vdash x \hookrightarrow \tau}
    \and
    \inferrule[$\Gamma$-Rd-Box]
    {\Gamma\vdash p \hookrightarrow \kwd{box}(\tau)}
    {\Gamma\vdash * p\hookrightarrow \tau}
    \and
    \inferrule[$\Gamma$-Rd-Shared]
    {\Gamma\vdash p \hookrightarrow \&^\kwd{s}(\tau)}
    {\Gamma\vdash *p \hookrightarrow \tau}
    \and
    \inferrule[$\Gamma$-Rd-Mutable]
    {\Gamma\vdash p \hookrightarrow \&^\kwd{m}(\tau_\text{c},\tau_\text{p})}
    {\Gamma\vdash * p\hookrightarrow \tau_\text{c}}    
    \end{mathpar}
    
    \judgement{Context Writing}{$\GWt{\Gamma}{p}{\tau}{\Gamma'}$}
    \begin{mathpar}
    \inferrule[$\Gamma$-Wt-Var]
    {\forall y\not=x, \Gamma'(y)=\Gamma(y) 
    \\ \Gamma'(x) = \tau}
    {\GWt{\Gamma}{x}{\tau}{\Gamma'}}
    \and
    \inferrule[$\Gamma$-Wt-Box]
    {\Gamma\vdash p\hookrightarrow \kwd{box}(\_)
    \\ \GWt{\Gamma}{p}{\kwd{box}(\kwd{list}(\alpha))}{\Gamma'}
    }
    {\GWt{\Gamma}{*p}{\kwd{list}(\alpha)}{\Gamma'}}
    \\
    \inferrule*[Right=\rulename{$\Gamma$-Wt-Shared}]
    {\Gamma\vdash p\hookrightarrow \&^\kwd{s}(\_)
    \\ \GWt{\Gamma}{p}{\&^\kwd{s}(\tau)}{\Gamma'}
    }
    {\GWt{\Gamma}{*p}{\tau}{\Gamma'}}
    \\
    \inferrule*[Right=\rulename{$\Gamma$-Wt-Mutable}]
    {\Gamma\vdash p\hookrightarrow \&^\kwd{m}(\tau_\text{c}, \tau_\text{p})
    \\ \vdash \tau_\text{c}
    \\ \GWt{\Gamma}{p}{\&^\kwd{m}(\tau, \tau_\text{p})}{\Gamma'}
    }
    {\GWt{\Gamma}{*p}{\tau}{\Gamma'}}
    \end{mathpar}
    \caption{Context Reading and Writing}
    \label{fig:sta-rw}
\end{figure}

\begin{figure}[t]
\small
    \judgement{Signatures}{$\vdash f \Rightarrow (\Gamma_f, \delta_f)$}
    \begin{mathpar}
    \inferrule
    {\text{fn}~ f ~(\vec{x}_\text{param}:\vec{t}_\text{param}, \vec{x}_\text{local}:\vec{t}_\text{local}, x_\text{ret}:t_\text{ret}) \{~ s ~\}
    \\\\ \textit{enrich}~ \vec{t}_\text{param} ~\textit{as}~ \vec{\tau}_\text{param}
    \\ \textit{enrich}~ \vec{t}_\text{local} ~\textit{as}~ \vec{\tau}_\text{local}
    \\ \textit{enrich}~ t_\text{ret} ~\textit{as}~ \tau_\text{ret}
    \\\\ \GWt{\emptyset}{\vec{x}_\text{param}}{\vec{\tau}_\text{param}}{\Gamma_1} 
    \\ \GWt{\Gamma_1}{\vec{x}_\text{local}}{\vec{\tau}_\text{local}}{\Gamma_2}
    \\ \GWt{\Gamma_2}{x_\text{ret}}{\tau_\text{ret}}{\Gamma_f}
    \\ \delta_f ~\text{fresh} }
    {\vdash f \Rightarrow (\Gamma_f, \delta_f)}
    \end{mathpar}
    \caption{Function Signatures}
    \label{fig:fun-sig}
\end{figure}

\textbf{Rich types} are types enriched with potential annotation $\alpha$ as in \cref{fig:rich-type} and \cref{fig:enrich}. 
The rich type $\bot$ denotes zero potential as the minimum among all rich types. 
The rich type $\kwd{list}(\alpha)$, represents the potential function $\alpha\cdot n$ for list $l$ with length $n$.
In shared borrows $\&^\kwd{s}(\tau)$, $\tau$ represents the potential function of borrowed value. 
Mutable borrows $\&^\kwd{m}(\tau_\text{c}, \tau_\text{p})$ contains 2 components. $\tau_\text{c}$ is the \textbf{c}urrent type, which denotes the current potential of mutable borrow. $\tau_\text{p}$ is the \textbf{p}rophecy type, which denotes the prophecy potential when the mutable borrow ends.  

Typing \textbf{context} $\Gamma : \mathbf{Variable}\to\mathbf{RichType}$ is a partial map, where unused variables can be mapped to $\bot$. Similarly, in \cref{fig:sta-rw}, we extend the reading and writing operation on typing context from variable $x$ to place $p$. It is worth noting that rules \rulename{$\Gamma$-Rd-Mutable} and \rulename{$\Gamma$-Wt-Mutable} indicate to read and write the mutable borrow on its current component $\tau_\text{c}$. We explicitly point out that $\vdash \tau_\text{c}$ in the premise of rule \rulename{$\Gamma$-Wt-Mutable} is \textbf{dropping condition} for soundness, detailed in \cref{sec:inference:subtyping}. Because when we update $\tau_\text{c}$, the old $\tau_\text{c}$ should be restored if it is a mutable borrow.

\textbf{Signature} $(\Sigma_f, \delta_f)$ of a function $f$ compose a typing context $\Sigma_f$ and a resource unknown variable $\delta_f \in \ZZ$. As shown in \cref{fig:fun-sig}, context $\Gamma_f$ contains rich types for parameters, local variables, and the return variable. $\delta_f$ indicates the resource consumption irrelevant to parameters.

\subsection{Subtyping, Well-formedness, and Merging} \label{sec:inference:subtyping}
\begin{figure}[t]
\small
    \judgement{Subtyping}{$\tau_1 \preceq \tau_2$}
    \begin{mathpar}
    \inferrule*[Right=\rulename{S-Bot}]
    {~}
    {\bot\preceq\tau}
    \and
    \inferrule*[Right=\rulename{S-Int}]
    {~}
    {\kwd{i32}\preceq\kwd{i32}}
    \and
    \inferrule*[Right=\rulename{S-Bool}]
    {~}
    {\kwd{bool}\preceq\kwd{bool}}
    \and
    \inferrule*[Right=\rulename{S-List}]
    {\alpha_1 \leq \alpha_2}
    {\kwd{list}(\alpha_1)\preceq\kwd{list}(\alpha_2)}
    \\
    \inferrule[S-Box]
    {\alpha_1 \leq \alpha_2}
    {\kwd{box}(\kwd{list}(\alpha_1))\preceq\kwd{box}(\kwd{list}(\alpha_2))}
    \and
    \inferrule[S-Shared]
    {\tau_1 \preceq \tau_2}
    {\&^\kwd{s}(\tau_1)\preceq\&^\kwd{s}(\tau_2)}
    \and
    \inferrule[S-Mutable]
    {\tau_{\text{c}, 1}\preceq \tau_{\text{c}, 2}
    \\ \tau_{\text{p}, 2}\preceq \tau_{\text{p}, 1} }
    {\&^\kwd{m}(\tau_{\text{c}, 1}, \tau_{\text{p}, 1})\preceq\&^\kwd{m}(\tau_{\text{c}, 2}, \tau_{\text{p}, 2})}
    \end{mathpar}
    \caption{Rich Subtyping}
    \label{fig:rich-subtyping}
\end{figure}
\begin{figure}[t]
\small
    \judgement{Well-formedness}{$\vdash \tau$}
    \begin{mathpar}

    \inferrule[WF-Bot]
    {~}
    {\vdash \bot}
    \and
    \inferrule[WF-Int]
    {~}
    {\vdash \kwd{i32}}
    \and
    \inferrule[WF-Bool]
    {~}
    {\vdash \kwd{bool}}
    \and
    \inferrule[WF-List]
    {\alpha \geq 0}
    {\vdash \kwd{list}(\alpha)}
    \and
    \inferrule[WF-Box]
    {\alpha \geq 0}
    {\vdash \kwd{box}(\kwd{list}(\alpha))}
    \\
    \inferrule*[Right=\rulename{WF-Shared}]
    {\vdash \tau
    }
    {\vdash \&^\kwd{s}(\tau)}
    \and
    \inferrule*[Right=\rulename{WF-Mutable}]
    {\tau_\text{p} \preceq \tau_\text{c}
    \\ \vdash \tau_\text{c}
    \\ \vdash \tau_\text{p}
    }
    {\vdash \&^\kwd{m}(\tau_\text{c}, \tau_\text{p})}
    \end{mathpar}
    \caption{Well-formedness}
    \label{fig:rich-type-wf}
\end{figure}
\begin{figure}[t]
\small
    \judgement{Context Merging}{$\Gamma_1 \sqcap \Gamma_2 = \{ x \hookrightarrow \Gamma_1(x)\cap\Gamma_2(x) : x \in \mathbf{dom}(\Gamma_1)=\mathbf{dom}(\Gamma_2)\}$}
    \judgement{Meet/Join (Selected)}{$\tau_1\cap\tau_2 / \tau_1\cup\tau_2$}
    \begin{mathpar}        
    \inferrule*[Right=Meet-List]
    {\min(\alpha_1, \alpha_2)=\alpha}
    {\kwd{list}(\alpha_1)\cap\kwd{list}(\alpha_2)=\kwd{list}(\alpha)}
    \and
    \inferrule*[Right=Join-List]
    {\max(\alpha_1, \alpha_2)=\alpha}
    {\kwd{list}(\alpha_1)\cup\kwd{list}(\alpha_2)=\kwd{list}(\alpha)}
    \\
    
    \inferrule*[Right=Meet-Shared]
    {\tau_1 \cap \tau_2=\tau}
    {\&^\kwd{s}(\tau_1)\cap\&^\kwd{s}(\tau_2)=\&^\kwd{s}(\tau)}
    \and
    \inferrule*[Right=Join-Shared]
    {\tau_1 \cup \tau_2=\tau}
    {\&^\kwd{s}(\tau_1)\cup\&^\kwd{s}(\tau_2)=\&^\kwd{s}(\tau)}
    \\

    \inferrule*[Right=Meet-Mutable]
    {\tau_{\text{c}, 1} \cap \tau_{\text{c}, 2}=\tau_\text{c}
    \\ \tau_{\text{p}, 1} \cup \tau_{\text{p}, 2}=\tau_\text{p}
    \\ \tau_{\text{p}, 1} \preceq \tau_{\text{c}, 1}
    \\ \tau_{\text{p}, 2} \preceq \tau_{\text{c}, 2}
    }
    {\&^\kwd{m}(\tau_{\text{c}, 1}, \tau_{\text{p}, 1})\cap\&^\kwd{m}(\tau_{\text{c}, 2}, \tau_{\text{p}, 2})=\&^\kwd{m}(\tau_\text{c}, \tau_\text{p})}
    \\
    \inferrule*[Right=Join-Mutable]
    {\tau_{\text{c}, 1} \cup \tau_{\text{c}, 2}=\tau_\text{c}
    \\ \tau_{\text{p}, 1} \cap \tau_{\text{p}, 2}=\tau_\text{p}
    \\ \tau_{\text{p}, 1} \preceq \tau_{\text{c}, 1}
    \\ \tau_{\text{p}, 2} \preceq \tau_{\text{c}, 2}
    }
    {\&^\kwd{m}(\tau_{\text{c}, 1}, \tau_{\text{p}, 1})\cup\&^\kwd{m}(\tau_{\text{c}, 2}, \tau_{\text{p}, 2})=\&^\kwd{m}(\tau_\text{c}, \tau_\text{p})}
    \end{mathpar}
    \caption{Merging}
    \label{fig:sta-merging}
\end{figure}

The order relation $\leq$ on resources derives another order relation on rich types, the \emph{subtyping} relation in \cref{fig:rich-subtyping}. The interpretation of subtyping $\tau_1 \preceq \tau_2$ is that the value $v$ typed with $\tau_1$ has \textbf{less} resource than the value $v$ typed with $\tau_2$. 
The rich type $\bot$ is a subtype of any type because $\bot$ denotes zero potential.
It is worth noting that in \rulename{S-Mutable}, $\tau_\text{p}$ is contravariant because prophecy type $\tau_\text{p}$ denotes the prophecy potential to return. 
The reflexive rule and the transitive rule are derivable.

A well-formed rich type always denotes a non-negative potential function. Our type system can drop well-formed types without sacrificing soundness.
\rulename{WF-List} and \rulename{WF-Box} request $\alpha \geq 0$, which makes $\alpha \cdot n \geq 0$ for list $l$ with length $n\geq0$. 
\rulename{WF-Shared} is a structural rule. For example, if $\kwd{list}(\alpha)$ is well-formed, so is $\&^\kwd{s}(\kwd{list}(\alpha))$. Rust borrow checker ensures that $\tau$ in $\&^\kwd{s}(\tau)$ satisfies $\tau\not=\&^\kwd{m}(\_, \_)$. Our type system supports nested borrows.
Besides structural premises $\vdash \tau_\text{c}$ and $\vdash \tau_\text{p}$, \rulename{WF-Mutable} demands \textbf{dropping condition} $\tau_\text{p} \preceq \tau_\text{c}$ in $\&^\kwd{m}(\tau_\text{c}, \tau_\text{p})$. The condition is called dropping condition because it works as dropping mutable borrows in \cref{fig:ex-prophecy}. The dropping condition makes sure that mutable borrow types denote non-negative potentials, as illustrated in \cref{sec:soundness}.

\textbf{Merging} is a conservative approximation of resource potentials after conditional branching. Under typing context $\Gamma$, the type system checks statements $s_1$ and $s_2$ in different branches and gets remainder contexts $\Gamma_1$ and $\Gamma_2$.
The type system should merge them to check continuation. 
As illustrated in \cref{fig:sta-merging}, to merge typing contexts is to merge rich types at each $x$ in the domain of two contexts. 
Because the prophecy type $\tau_\text{p}$ in mutable borrow is contravariant, we need to define not only the meet of types but also the join of types. Our purpose is to construct a \emph{lattice} with the property that $\tau_1\cap\tau_2\preceq \tau_i \preceq\tau_1\cup\tau_2, \forall i=1, 2$. Hence, merging over contexts is non-increasing on resources to conservatively fulfill soundness. 
The lattice operations of $\kwd{list}(\alpha_1)$ and $\kwd{list}(\alpha_1)$ are inherited from the resource's $\min$ and $\max$, so it is readily comprehensible. 
Notice that dropping conditions appear in rules \rulename{Meet-Mutable} and \rulename{Join-Mutable}. They are to fulfill soundness for weak updates, which is mentioned in \cref{sec:overview:Lattice}. Recall that dropping borrows without these conditions may increase resources in both original places indicated by $\tau_{\text{p}, 1}$ and $\tau_{\text{p}, 2}$, to break soundness. 


It is worth noting that we support nested borrows like $\&^\kwd{s}(\&^\kwd{s}(\tau))$, $\&^\kwd{m}(\&^\kwd{s}(\tau_\text{c}), \&^\kwd{s}(\tau_\text{p}))$ and $\&^\kwd{m}(\&^\kwd{m}(\tau_\text{cc}, \tau_\text{cp}), \&^\kwd{m}(\tau_\text{cc}, \tau_\text{pp}))$. Rust's borrow mechanisms exclude nested borrows like shared borrows of mutable borrows $\&^\kwd{s}(\&^\kwd{m}(\tau_\text{c}, \tau_\text{p}))$, because they violate the property that at most one mutable borrow from the same piece of data is live at the same time.

\subsection{Typing Expressions} \label{sec:inference:eval}\
\begin{figure}[t]
\small
\judgement{Typing Expressions (Selected)}{$\Gamma\vdash e \hookrightarrow \tau\dashv\Gamma'$}
    \begin{mathpar}
    \inferrule*[Right=\rulename{$\Gamma$-Ev-Nil}]
    {\alpha ~\text{fresh}}
    {\Gamma\vdash \kwd{nil} \hookrightarrow \kwd{list}(\alpha)\vdash\Gamma}
    \and
    \inferrule*[Right=\rulename{$\Gamma$-Ev-Move}]
    {\Gamma\vdash p \hookrightarrow \tau
    \\ \GWt{\Gamma}{p}{\bot}{\Gamma'}
    }
    {\Gamma\vdash \kwd{move}~p \hookrightarrow \tau\dashv\Gamma'}
    \\

    \inferrule*[Right=\rulename{$\Gamma$-Ev-Shared}]
    {\Gamma\vdash p \hookrightarrow \tau
    \\ \textit{share}~ \tau ~\textit{as}~\tau_1, \tau_2
    \\ \GWt{\Gamma}{p}{\tau_1}{\Gamma'}
    }
    {\Gamma\vdash \&^\kwd{s}~p \hookrightarrow \&^\kwd{s}(\tau_2)\dashv\Gamma'}
    \\
    
    \inferrule*[Right=\rulename{$\Gamma$-Ev-Mutable}]
    {\Gamma\vdash p \hookrightarrow \tau
    \\ \textit{prophesy}~ \tau ~\textit{as}~ \tau_\text{p} 
    \\ \GWt{\Gamma}{p}{\tau_\text{p}}{\Gamma'}
    }
    {\Gamma\vdash \&^\kwd{m}~p \hookrightarrow \&^\kwd{m}(\tau, \tau_\text{p})\dashv\Gamma'}
    \end{mathpar}
    \caption{Typing Expressions}
    \label{fig:sta-eval}
\end{figure}

\begin{figure}[t]
\small
    \judgement{Sharing (Selected)}{$\textit{share}~ \tau ~\textit{as}~\tau_1, \tau_2$}
    \begin{mathpar}
    \inferrule*[Right=\rulename{Share-List}]
    {\alpha_1, \alpha_2 ~\text{fresh}
    \\\alpha = \alpha_1 + \alpha_2}
    {\textit{share}~ \kwd{list}(\alpha) ~\textit{as}~\kwd{list}(\alpha_1), \kwd{list}(\alpha_2)}
    \end{mathpar}
    \judgement{Prophesying (Selected)}{$\textit{prophesy}~ \tau_\text{c} ~\textit{as}~\tau_\text{p}$}
    \begin{mathpar}
    \inferrule*[Right=\rulename{Prophesy-List}]
    {\alpha_\text{p}~\text{fresh}}
    {\textit{prophesy}~ \kwd{list}(\alpha) ~\textit{as}~ \kwd{list}(\alpha_\text{p})}
    \end{mathpar}
    \caption{Sharing and Prophesying}
    \label{fig:sta-sharing-prophesying}
\end{figure}

\cref{fig:sta-eval} presents how to type check expressions via judgement $\Gamma\vdash e\hookrightarrow \tau\dashv\Gamma'$. Unlike the dynamic evaluation $V\vdash e \rightsquigarrow v$, checking expressions may modify $\Gamma$ to the remainder context $\Gamma'$. 
Rule \rulename{$\Gamma$-Ev-Nil} introduces a fresh unknown potential annotation $\alpha$ for $\kwd{nil}$. 
Rule \rulename{$\Gamma$-Ev-Move} explicitly moves the type $\tau$ out from place $p$, making $\GWt{\Gamma}{p}{\bot}{\Gamma'}$. 

Shared and mutable borrows modify typing context, as illustrated in rule \rulename{$\Gamma$-Ev-Shared} and \rulename{$\Gamma$-Ev-Mutable} with sharing and prophesying. \cref{fig:sta-sharing-prophesying} selects essential rules of $\textit{share}~ \tau ~\textit{as}~ \tau_1, \tau_2$ and $\textit{prophesy}~ \tau ~\textit{as}~ \tau_\text{p}$ for borrows. 

\textbf{Shared borrows} are handled with sharing $\textit{share}~ \tau ~\textit{as}~ \tau_1, \tau_2$. Recall the example in \cref{fig:ex-sharing}. We select the rule \rulename{Share-List} to reveal the essence of sharing. Sharing is splitting resource annotation $\alpha$ into $\alpha_1$ and $\alpha_2$ with linear constraint $\alpha = \alpha_1 + \alpha_2$. In rule \rulename{$\Gamma$-Ev-Shared}, we write $\tau_1$ back to original place $p$, with $\tau_2$ lent out. There is no sharing of mutable borrows as $\textit{share}~ \&^\kwd{m}(\_, \_) ~\textit{as}~ \tau_1, \tau_2$, because a well-checked program will never incur shared borrows of mutable borrows $\&^\kwd{s}(\&^\kwd{m}(\tau_\text{c}, \tau_\text{p}))$.

\textbf{Mutable borrows} are handled with prophesying $\textit{prophesy}~ \tau ~\textit{as}~ \tau_\text{p}$. Recall the example in \cref{fig:ex-prophecy}. The selected rule \rulename{Prophesy-List} prophesy $\alpha_\text{p}$ as the prophecy potential when the mutable borrow ends. In rule \rulename{$\Gamma$-Ev-Mutable}, we write prophecy type $\tau_\text{p}$ to the place $p$. Once the borrow ends, the dropping condition $\vdash \&^\kwd{m}(\tau, \tau_\text{p})$ ensures that the prophecy type $\tau_\text{p}$ is bounded by current type $\tau$.

\subsection{Typing Statements} \label{sec:inference:exec}
\begin{figure}[t]
\small
    \judgement{Typing Statements (Selected)}{$\Gamma\vdash s \hookrightarrow^\delta \dashv\Gamma'$}
    \begin{mathpar}
    \inferrule*[Right=\rulename{$\Gamma$-Ex-Tick}]
    {~}
    {\Gamma\vdash\kwd{tick}(\delta)\hookrightarrow^\delta\vdash\Gamma}
    \and
    \inferrule*[Right=\rulename{$\Gamma$-Ex-Drop}]
    {\Gamma\vdash p\hookrightarrow \tau
    \\ \vdash \tau
    \\ \GWt{\Gamma}{p}{\bot}{Gamma'}
    }
    {\Gamma\vdash \kwd{drop}~p \hookrightarrow^0\dashv \Gamma'}
    \\

    \inferrule*[Right=\rulename{$\Gamma$-Ex-Cons}]
    {\Gamma\vdash e_1\hookrightarrow \kwd{i32} \dashv \Gamma_1
    \\ \Gamma_1\vdash e_2\hookrightarrow \kwd{box}(\kwd{list}(\alpha'))\dashv\Gamma_2
    \\ \GWt{\Gamma_2}{p}{\kwd{list}(\alpha')}{\Gamma'}}
    {\Gamma\vdash p\from \kwd{cons}(e_1, e_2)\hookrightarrow^{\alpha'}\dashv\Gamma'}
    \\

    \inferrule*[Right=\rulename{$\Gamma$-Ex-If}]
    {\Gamma\vdash p\hookrightarrow \kwd{bool}
    \\ \Gamma\vdash s_1\hookrightarrow^{\delta_1}\dashv\Gamma_1
    \\ \Gamma\vdash s_2\hookrightarrow^{\delta_2}\dashv\Gamma_2
    \\ \max(\delta_1, \delta_2)=\delta
    \\ \Gamma_1\sqcap\Gamma_2=\Gamma' }
    {\Gamma\vdash \kwd{if}~ p ~\kwd{then}~ s_1 ~\kwd{else}~ s_2 ~\kwd{end} \hookrightarrow^\delta \dashv\Gamma'}
    \\
    \inferrule*[Right=\rulename{$\Gamma$-Ex-Mat}]
    {\Gamma\vdash p\hookrightarrow \kwd{list}(\alpha)
    \\ \Gamma\vdash s_1\hookrightarrow^{\delta_1}\dashv\Gamma_1
    \\\\ \GWt{\Gamma}{p}{\bot}{\Gamma_{\text{b}, 1}}
    \\ \GWt{\Gamma_{\text{b}, 1}}{x_\text{hd}}{\kwd{i32}}{\Gamma_{\text{b}, 2}}
    \\ \GWt{\Gamma_{\text{b}, 2}}{x_\text{tl}}{\kwd{box}(\kwd{list}(\alpha))}{\Gamma_\text{b}}
    \\ \Gamma_\text{b}\vdash s_2\hookrightarrow^{\delta_2}\dashv\Gamma'_\text{b}
    \\\\ \Gamma'_\text{b}\vdash x_\text{tl}\hookrightarrow \kwd{list}(\beta)
    \\ \GWt{\Gamma'_\text{b}}{x_\text{hd}}{\bot}{\Gamma'_{\text{b}, 1}}
    \\ \GWt{\Gamma'_{\text{b}, 1}}{x_\text{tl}}{\bot}{\Gamma'_{\text{b}, 2}}
    \\ \GWt{\Gamma'_{\text{b}, 2}}{p}{\kwd{list}(\beta)}{\Gamma_2}
    \\\\ \max(\delta_1, \delta_2-(\alpha-\beta))=\delta
    \\ \Gamma_1\sqcap\Gamma_2=\Gamma'}
    {\Gamma\vdash \kwd{match}~ p ~ \{\kwd{nil}\mapsto s_1, \kwd{cons}(x_\text{hd}, x_\text{tl})\mapsto s_2\} \hookrightarrow^\delta \dashv\Gamma'}
    \\

    \inferrule*[Right=\rulename{$\Gamma$-Ex-App}]
    {\text{fn}~ f ~(\vec{x}_\text{param}:\vec{t}_\text{param}, \vec{x}_\text{local}:\vec{t}_\text{local}, x_\text{ret}:t_\text{ret}) \{~ s ~\}
    \\\\ \vdash f \Leftarrow (\Gamma_f, \delta_f)
    \\ \Gamma_f\vdash x_\text{ret} \hookrightarrow \tau_\text{ret}, (\forall x_i\in\vec{x}_\text{param}, i=1, ..., n) \Gamma_f \vdash x_i \hookrightarrow \tau_{\text{param}, i}
    \\\\ \Gamma_0=\Gamma, (\forall e_i\in \vec{e}, i=1, ..., n) \Gamma_{i-1}\vdash e_i\hookrightarrow\tau_{\text{arg}, i}\dashv\Gamma_i
    \\ (\forall i=1,..,n)~ \tau_{\text{param}, i} = \tau_{\text{arg}, i}
    \\ \Gamma_n \vdash p \hookrightarrow \tau
    \\ \vdash \tau
    \\ \GWt{\Gamma_n}{p}{\tau_\text{ret}}{\Gamma'}
    }
    {\Gamma\vdash p\from f(\vec{e})\hookrightarrow^{\delta_f}\dashv\Gamma'}
    \end{mathpar}
    \caption{Typing Statements}
    \label{fig:sta-exec}
\end{figure}

\begin{figure}[t]
\small
    \judgement{Function Analysis}{$\vdash f \Leftarrow (\Gamma_f, \delta_f)$}
    \begin{mathpar}
    \inferrule
    {\text{fn}~ f ~(\vec{x}_\text{param}:\vec{t}_\text{param}, \vec{x}_\text{local}:\vec{t}_\text{local}, x_\text{ret}:t_\text{ret}) \{~ s ~\}
    \\  \vdash f \Rightarrow (\Gamma_f, \delta_f)
    \\ \Gamma_f\vdash s\hookrightarrow^\delta\dashv\Gamma'_f
    \\\\ \forall x \in \textbf{dom}(\Gamma'_f), \vdash \Gamma'_f(x)
    \\ \Gamma'_f \vdash x_\text{ret} \hookrightarrow \tau'_\text{ret}
    \\ \Gamma_f \vdash x_\text{ret} \hookrightarrow \tau_\text{ret}
    \\ \tau'_\text{ret} = \tau_\text{ret}
    \\ \delta = \delta_f}
    {\vdash f \Leftarrow (\Gamma_f, \delta_f)}
    \end{mathpar}
    \caption{Function Analysis}
    \label{fig:fun-anal}
\end{figure}

\cref{fig:sta-exec} presents how to type check statements as judgement $\Gamma\vdash s \hookrightarrow^\delta \dashv\Gamma'$. Under context $\Gamma$, the statement $s$ is checked with resource consumption $\delta$, and context becomes $\Gamma'$. 

Rule \rulename{$\Gamma$-Ex-Tick} indicates $\kwd{tick}(\delta)$ consumes $\delta$ unit of resource. Rule \rulename{$\Gamma$-Ex-Drop} drops the type $\tau$ with well-formedness $\vdash\tau$ as the dropping condition. Rule \rulename{$\Gamma$-Ex-Cons} indicates that $\kwd{cons}$ will consume $\alpha$ unit of resource for continuation payment, when the tail $e_2$ is typed with $\kwd{box}(\kwd{list}(\alpha))$. 

\textbf{Branching statements} require context merging, as in \rulename{$\Gamma$-Ex-If} and \rulename{$\Gamma$-Ex-Mat}. Rule \rulename{$\Gamma$-Ex-If} is simpler to merge contexts with the consumption as the maximum of those branches. Rule \rulename{$\Gamma$-Ex-Mat} is more intricate, due to resource potential stored in \kwd{cons}. The $\kwd{cons}$ branch will obtain $\alpha-\beta$ units of potential, therefore the net consumption is $\delta_2-(\alpha-\beta)$. Given $\Gamma\vdash p \hookrightarrow \kwd{list}(\alpha)$, The potential is not $\alpha$ but $\alpha-\beta$. $\beta$ is the remainder potential, indicated by $\Gamma'_\text{b} \vdash x_\text{tl} \hookrightarrow \kwd{list}(\beta)$. The subscript $\text{b}$ of $\Gamma_\text{b}$ means \textbf{b}inding, similar to rule \rulename{V-Ex-Mat}.

\textbf{Function application} is intractable because of recursive functions. Rule \rulename{$\Gamma$-Ex-App} assumes the function $f$ has a well-checked signature $(\Gamma_f, \delta_f)$, with judgement $\vdash f \Leftarrow (\Gamma_f, \delta_f)$ in \cref{fig:fun-anal}, different from $\vdash f \Rightarrow (\Gamma_f, \delta_f)$. Other premises are to ensure that the resources of actual arguments are equal to those of formal parameters. 

\subsection{Type Inference} \label{sec:inference:infer}
To this point, our type system has been primarily declarative because the well-checked signature in rule \rulename{$\Gamma$-Ex-App} is assumed to be pre-existent. Same as other AARA systems (such as Resource-aware ML~\cite{RaML}), we use linear programming to convert the declarative type system to an algorithmic version. The type system creates symbolic variables to denote unknown annotations in rich types and signatures. The type system then collects linear constraints among those symbolic variables and finally solves them via linear programming solvers. 

Readers might have perceived that a recursive function requires a well-checked signature during checking and that a function can exhibit multiple signatures at different call sites. To automatically analyze functions, we need to preprocess the call graph. First, we group recursive functions as strongly connected components. Second, we topologically sort groups to determine an order to analyze. For each group, we predefine signatures of functions in the group via the judgement $\textit{enrich}~ t ~\textit{as}~ \tau$ in \cref{fig:enrich}. During function analysis, the signature $(\Gamma_f, \delta_f)$ in \rulename{$\Gamma$-Ex-App} should be replaced with the predefined one if $f$ is in the group. Otherwise, $f$ is in the previously analyzed group, so we should clone that group's signature and linear constraints. It is necessary to clone instead of copy them because annotations in signatures and constraints are sensitive to actual arguments of function calls.

With linear constraints collected during function analysis and a heuristic objective, we can employ a linear programming solver to find instances of annotations that satisfy those constraints automatically. The inferred annotations in signatures will characterize functions' resource consumption. 

\section{Soundness} \label{sec:soundness}

In this section, we define potential functions indicated by the resource-enriched types from \cref{sec:inference} in \cref{sec:pot-funcs} and then sketch the soundness proof of the resource-aware type system in \cref{sec:proof-sketch}. We include the detailed proofs in the \cref{sec:proof}.

\subsection{Potential Functions}
\label{sec:pot-funcs}
 

\begin{figure}[t]
\small
    \judgement{Potential Functions}{$\Phi(V:\Gamma) = \sum_{x\in\textbf{dom}(\Gamma)} \phi(V(x):\Gamma(x))$} \\
    \judgement{Potential Functions (Selected)}{$\phi(v:\tau)$}
    \begin{mathpar}
    \inferrule[$\phi$-Nil]
    {~}
    {\phi(\kwd{nil}:\kwd{list}(\alpha))=0}
    \and
    \inferrule[$\phi$-Cons]
    {~}
    {\phi(\kwd{cons}(\kwd{n}_\text{i32}, \kwd{box}(lv)):\kwd{list}(\alpha))=\alpha+\phi(\kwd{box}(lv):\kwd{box}(\kwd{list}(\alpha)))}
    \\
    \inferrule[$\phi$-Shared]
    {~}
    {\phi(\&(\_, v):\&^\kwd{s}(\tau))=\phi(v:\tau)}
    \and
    \inferrule[$\phi$-Mutable]
    {~}
    {\phi(\&(\_, v):\&^\kwd{m}(\tau_\text{c}, \tau_\text{p}))=\phi(v:\tau_\text{c})-\phi(v:\tau_\text{p})}
    \end{mathpar}
    \caption{Potential Functions}
    \label{fig:sound-potential}
\end{figure}

\cref{fig:sound-potential} defines potential function $\phi(v:\tau)$ and $\Phi(V:\Gamma)$.
\rulename{$\phi$-Nil} and \rulename{$\phi$-Cons} define the potential of a list $l : \kwd{list}(\alpha)$ with length $n$ to be $\alpha \cdot n$.
\rulename{$\phi$-Shared} defines the potential of shared borrows to be the potential of borrowed values and borrowed rich types. 
\rulename{$\phi$-Mutable} defines the potential of mutable borrows to be the difference between current and prophecy potential. It is worth noting that when the program incurs a mutable borrow, the potential of context will not change. The dropping condition $\tau_\text{p} \preceq \tau_\text{c}$ in $\vdash \&^\kwd{m}(\tau_\text{c}, \tau_\text{p})$, ensures the potential is non-negative. We, therefore have the following lemma about potential:

\begin{lemma}
    Potential is non-negative and keeps subtyping:
    \begin{mathpar}
    \inferrule
    {  \vdash \tau_1
    \\ \vdash \tau_2
    \\ \tau_1 \preceq \tau_2
    }
    { 0 \leq \phi(v:\tau_1) \leq \phi(v:\tau_2)}
    \end{mathpar}
\end{lemma}
\begin{proof}
    First define the size of rich types structural inductively $\mathbf{size} : \mathbf{RichType} \to \NN$, \ 
    and prove by induction on $\mathbf{size}(\tau_1) + \mathbf{size}(\tau_2)$. The well-formedness will be used to prove potential of mutable borrows is non-negative.
\end{proof}
\begin{corollary}
    Potential is non-negative $\dfrac{\vdash \tau}{0\leq \phi(v:\tau)}$ due to the derivable rule \rulename{S-Refl} $\dfrac{~}{\tau\preceq\tau}$.
\end{corollary}

The subtyping relation and well-formedness can be extended to typing context, which can be conceptualized as record types. Lemmas about the lattice operation and potential can be proved with definition and simple induction.

\begin{definition}
    Subcontext $\Gamma_1 \subseteq \Gamma_2$ if and only if $\forall x \in \textbf{dom}(\Gamma_1), x \in \textbf{dom}(\Gamma_2), \Gamma_1(x)\preceq \Gamma_2(x)$. \\
    Context well formed $\vdash \Gamma$ if and only if $\forall x \in \textbf{dom}(\Gamma), \vdash \Gamma(x)$.
\end{definition}
\begin{lemma}[]
    For any store $V$ and any context $\Gamma_1, \Gamma_2$, 
    \begin{enumerate}
        \item {$\Gamma_1\sqcap\Gamma_2 \subseteq \Gamma_1$ and $\Gamma_1\sqcap\Gamma_2 \subseteq \Gamma_2$; }
        \item {if $\vdash \Gamma_1, \vdash \Gamma_2$, then $\vdash \Gamma_1\sqcap\Gamma_2$; }
        \item {if $\vdash \Gamma_1, \vdash \Gamma_2, \Gamma_1 \subseteq \Gamma_2$, then $0\leq \Phi(V:\Gamma_1) \leq \Phi(V:\Gamma_2)$. }
    \end{enumerate}
\end{lemma}

\subsection{Soundness Theorem}
\label{sec:proof-sketch}

With potential, we are able to formulate soundness theorem, which states that resource consumption of dynamics, together with potential difference, is bounded by type system.
\begin{theorem}[Soundness]
Our type system is sound towards resource aware dynamic semantics: \\
If $V\vDash s \rightsquigarrow^{\delta_V} \Dashv V'$ and $\Gamma \vdash s \hookrightarrow^{\delta_\Gamma} \dashv \Gamma'$, then $\Phi(V':\Gamma') - \Phi(V:\Gamma)+\delta_V \leq \delta_\Gamma$.
\end{theorem}
\begin{proof}
By induction on $V\vDash s \rightsquigarrow^{\delta_V} \Dashv V'$, with the help of following lemmas.
\end{proof}


\begin{lemma}[Update]
If store or context is written with new value or new type, the difference of potential over store or context is equal to that over new value or new type.
\begin{enumerate}
    \item {If $V\vDash p\rightsquigarrow v$, $\Gamma\vdash p\hookrightarrow \tau$ and $\VWt{V}{p}{v'}{V'}$,
    then $\Phi(V':\Gamma) - \Phi(V:\Gamma)=\Phi(v':\tau)-\Phi(v:\tau)$;}
    \item {If $V\vDash p\rightsquigarrow v$, $\Gamma\vdash p\hookrightarrow \tau$ and $\GWt{\Gamma}{p}{\tau'}{\Gamma'}$, 
    then $\Phi(V:\Gamma')-\Phi(V:\Gamma)=\Phi(v:\tau')-\Phi(v:\tau)$.}
\end{enumerate}
\end{lemma}
\begin{proof}
By induction on $\VWt{V}{p}{v'}{V'}$ and $\GWt{\Gamma}{p}{\tau'}{\Gamma'}$.
\end{proof}
\begin{lemma}[Evaluation]
If $V\vDash e\rightsquigarrow v, \Gamma\vdash e\hookrightarrow \tau\dashv\Gamma'$, then $\Phi(V:\Gamma')-\Phi(V:\Gamma)= -\phi(v:\tau)$.
\end{lemma}
\begin{proof}
By induction on $e$.
\end{proof}

Towards soundness proof for statements, especially for the rule \rulename{$\Gamma$-Ex-App}, we will need one series of lemmas about context extension. Though similar, context extension is different from subcontext. Context extension necessitates strict type equality, whereas subcontext only demands subtyping. The weakening rules are intuitive about context extension, as they merely involve adding variables and types that the program will not utilize. 

\begin{definition}
    Extension $\Gamma_1\sqsubseteq\Gamma_2$ if and only if $\forall x \in\textbf{dom}(\Gamma_1), x\in\textbf{dom}(\Gamma_2), \Gamma_1(x) = \Gamma_2(x)$.
\end{definition}
\begin{lemma}[Weakening]
In following rules, $\dfrac{A ~ B}{C ~ D}$ means if $A$ and $B$ then $C$ and $D$.
\begin{mathpar}
    \inferrule[$\Gamma$-Rd-Weaken]
    {\Gamma_1 \sqsubseteq \Gamma_2 
    \\ \Gamma_1 \vdash p \hookrightarrow \tau_1
    }
    {\Gamma_2 \vdash p \hookrightarrow \tau_2
    \\ \tau_1 = \tau_2
    }
    \and
    \inferrule[$\Gamma$-Wt-Weaken]
    {\Gamma_1 \sqsubseteq \Gamma_2
    \\ \GWt{\Gamma_1}{p}{\tau_1}{\Gamma'_1}
    \\ \tau_1 = \tau_2
    }
    { \GWt{\Gamma_2}{p}{\tau_2}{\Gamma'_2}
    \\ \Gamma'_1 \sqsubseteq \Gamma'_2 
    }
    \\
    \inferrule[$\Gamma$-Ev-Weaken]
    {\Gamma_1 \sqsubseteq \Gamma_2
    \\ \Gamma_1 \vdash e \hookrightarrow \tau_1 \dashv \Gamma'_1 
    }
    {\Gamma_2 \vdash e \hookrightarrow \tau_2 \dashv \Gamma'_2
    \\ \tau_1 = \tau_2 
    \\ \Gamma'_1 \sqsubseteq \Gamma'_2
    }
    \and
    \inferrule[$\Gamma$-Ex-Weaken]
    {\Gamma_1 \sqsubseteq \Gamma_2
    \\ \Gamma_1 \vdash s \hookrightarrow^{\delta_1} \dashv \Gamma'_1
    }
    { \Gamma_2 \vdash s \hookrightarrow^{\delta_2} \dashv \Gamma'_2
    \\ \delta_1 = \delta_2
    \\ \Gamma'_1 \sqsubseteq \Gamma'_2
    }
\end{mathpar}
\end{lemma}
\begin{proof}
    By induction on $\Gamma_1\vdash p\hookrightarrow \tau_1$, $\GWt{\Gamma_1}{p}{\tau_1}{\Gamma'_1}$, $\Gamma_1\vdash e\hookrightarrow\tau_1\dashv\Gamma'_1$, $\Gamma_1\vdash s\hookrightarrow^{\delta_1}\dashv\Gamma'_1$.
\end{proof}

\section{Experimental Evaluation} \label{sec:impl}

In this section, we present an experimental evaluation of \rarust{}.
\cref{sec:proto} describes our prototype implementation of \rarust{}.
\cref{sec:eval} presents the evaluation results of \rarust{} on a suite of benchmarks.

\subsection{Implementation}
\label{sec:proto}


We have implemented a prototype linear resource analyzer, \rarust{}. It takes raw Rust programs (within only $\kwd{tick}$ annotation) as inputs, and prints functions' signatures with resource annotation as output.
\rarust{} analyzes the whole program regardless of whether there are annotations or not. We currently use explicit manually annotated $\kwd{tick}$ as the cost model, but it is straightforward to support parametric cost models, which assign a cost to each kind of statements, as prior AARA systems (e.g., \textsc{RaML}~\cite{RaML}) do.
(1) \rarust{} first obtains the typed IR of the borrow calculus with explicit drops via Charon \cite{Aeneas} as a plugin of Rust compiler. Rust compiler guarantees that the compiled IR is well-checked, and \rarust{} will utilize properties straightforwardly. 
(2) Towards IR, \rarust{} analyzes strongly connected components of function call graph and topologically sort components, generating precedence of type checking. 
(3) \rarust{} enriches function signatures with fresh linear variables as resource annotation and assigns each component a linear programming context to record linear constraints. 
(4) As stated in \cref{sec:inference:infer}, \rarust{} type-checks functions' bodies, generating linear constraints.
(5) \rarust{} finally solves these constraints, with its heuristic linear objective, by invoking linear programming solvers. We can utilize different solvers due to the separation of (4) and (5), and we select CoinCbc \cite{CoinCbc} for demonstration.

Our implementation supports some features, based on core calculus formalized in \cref{sec:calculus} and some practical extensions. We list as follows:
\begin{enumerate}
    \item{
    \rarust{} supports reborrows and nested borrows like \lstinline|& & T|, \lstinline|&mut & T| and \lstinline|&mut &mut T|.
    }
    \item {
    \rarust{} supports user-defined data types instead of only built-in lists, as \cref{tab:user-defined} shows. 
    We annotate potential for each constructor, \lstinline|Nil|, \lstinline|Cons|, \lstinline|Leaf|, \lstinline|Node|, \lstinline|(_, _)|, \lstinline|Record(_, _)|, \lstinline|NListNode| and \lstinline|NList|. It is worth noting that \lstinline|NListNode| and \lstinline|NList| are mutually recursive.
    }
    \item {
    \rarust{} supports looping statements including $\kwd{while true}(s), \kwd{break}(i), \kwd{continue}(i)$, where $i$ represents the $i$-th loop outward from $\kwd{break}$ or $\kwd{continue}$. \rarust{} supports this directly without transforming loops into recursive functions.
    }
\end{enumerate}

\begin{DIFnomarkup}
\begin{table}[t]
\centering
\caption{User-defined Data Types} \label{tab:user-defined}
\footnotesize
\begin{tabular}{|l|l|}
\hline
\begin{lstlisting}[language=Rust, style=colouredRust]
pub enum List {
    Nil,
    Cons(i32, Box<List>),
}
\end{lstlisting} 
&
\begin{lstlisting}[language=Rust, style=colouredRust]
pub enum Tree {
    Leaf,
    Node(i32, Box<Tree>, Box<Tree>),
}
\end{lstlisting} 
\\
\hline
\begin{lstlisting}[language=Rust, style=colouredRust]
pub type Tuple = (List, Tree);
\end{lstlisting}
&
\begin{lstlisting}[language=Rust, style=colouredRust]
pub struct Record{pub l:List, pub t:Tree}
\end{lstlisting}
\\
\hline
\begin{lstlisting}[language=Rust, style=colouredRust]
pub struct NListNode {
    pub value : i32,
    pub next : NList
}
\end{lstlisting} 
&
\begin{lstlisting}[language=Rust, style=colouredRust]
pub enum NList {
    None,
    Some(Box<NListNode>)
}
\end{lstlisting} 
\\
\hline
\end{tabular}
\end{table}
\end{DIFnomarkup}

\subsection{Evaluation}
\label{sec:eval}

\begin{DIFnomarkup}
\begin{table}[t]
    \centering
    \caption{Benchmarks}
    \label{tab:eval}
    \small
    \scalebox{0.67}{
    \begin{tabular}{|c|c|c|c|c|c|}
    \hline
    case & type & description & complexity & lines of code & size of constraints  \\
    \hline
    \multicolumn{6}{|c|}{\textbf{feature(s): mutable borrows}} \\ \hdashline
    end\_m & {\lstinline|fn(l:&m List)->&m List|} & find the mutable borrow of the {\lstinline|Nil|} of a list & $1+3|l|$ & 14 & 71 \\
    end\_c & {\lstinline|fn(l:&m List,o:&m&m List)|} & c style end\_m to write to {\lstinline|o|} & $2+3|l|$ & 13 & 96 \\ 
    append & {\lstinline|fn(l:&m List,x:i32)|} & append $x$ to $l$ by {\lstinline|end_m|} & $5+3|l|$ & 5 & 89 \\
    concat & {\lstinline|fn(l1:&m List,l2:List)|} & $l'_1 = l_1@l_2$ by {\lstinline|end_c|} & $6+3|l_1|$ & 8 & 144 \\
    \hline
    \multicolumn{6}{|c|}{\textbf{feature(s): shared borrows, mutable borrows, recursive functions, and loop statements}} \\ \hdashline
    sum\_rec & {\lstinline|fn(l:& List)->i32|} & sum up integers, recursively & $1 + 6|l|$ & 14 & 19 \\
    sum\_loop & {\lstinline|fn(l:& List)->i32|} & sum up integers, via loops & $1 + 6|l|$ & 24 & 54 \\
    
    rev\_rec & {\lstinline|fn(l:&m List, r: List)->List|} & reverse $l$ to head of $r$ & $1 + 9|l|$ & 15 & 57 \\
    
    rev\_loop & {\lstinline|fn(l:&m List)|} & reverse $l$ mutably via loops & $1 + 9|l|$ & 21 & 164 \\
    \hline
    \multicolumn{6}{|c|}{\textbf{feature(s): function calls}} \\ \hdashline
    sum2 & {\lstinline|fn(l:& List)|} & {\lstinline|sum_rec(l);sum_loop(l)|} & $2 + 12|l|$ & 4 & 82 \\
    rev & {\lstinline|fn(l:&m List)|} & reverse $l$ mutably via {\lstinline|rev_rec|} & $4 + 9|l|$ & 5 & 71 \\
    rev2 & {\lstinline|fn(l:&m List)|} & apply {\lstinline|rev|} to $l$ twice & $8 + 18|l|$ & 4 & 170 \\
    dup & {\lstinline|fn(l: List)->List|} & duplicate each element in $l$ & $1+11|l|$ & 16 & 37 \\
    dup2 & {\lstinline|fn(l: List)->List|} & apply {\lstinline|dup|} to $l$ twice & $2+33|l|$ & 4 & 81 \\
    \hline
    \multicolumn{6}{|c|}{\textbf{feature(s): reborrow and nested borrows}} \\ \hdashline
    reborrow\_s & {\lstinline|fn(l:& List)|} & reborrow $l$ as $ll$, {\lstinline|sum2(ll);sum2(l)|}& $4+24|l|$ & 6 & 176 \\
    reborrow\_m & {\lstinline|fn(l:&mut List)|} & reborrow $l$ as $ll$, {\lstinline|rev2(ll);rev2(l)|}& $16+36|l|$ & 6 & 364 \\
    nested\_s\_s & {\lstinline|fn(l:& & List)|} & {\lstinline|sum2(*l);|} & $2+12|l|$ & 4 & 88 \\
    nested\_m\_s & {\lstinline|fn(l:&m & List)|} & {\lstinline|sum2(*l);|} & $2+12|l|$ & 4 & 90 \\
    nested\_m\_m & {\lstinline|fn(l:&m &m List)|} & {\lstinline|rev2(*l);|} & $8+18|l|$ & 4 & 188 \\
    \hline
    \multicolumn{6}{|c|}{\textbf{feature(s): user-defined datatypes}} \\ \hdashline
    min & {\lstinline|fn(t:&Tree,d:i32)->i32|} & find min in $t$, $d$ as default & $1+5|t|$ & 14 & 19 \\
    max & {\lstinline|fn(t:&Tree,d:i32)->i32|} & find max in $t$, $d$ as default & $1+5|t|$ & 14 & 19 \\
    find & {\lstinline|fn(t:&Tree,x:i32)->bool|} & find whether $x$ is in $t$ & $1+8|t|$ & 29 & 31 \\
    add & {\lstinline|fn(t:&m Tree,x:i32)|} & add up $x$ to each element of $t$ & $1+10|t|$ & 15 & 53\\
    tuple & {\lstinline|fn(x:&m Tuple)|} & {\lstinline|rev2(x.0);min(x.1, 0);|} & $9+18|x.0|+5|x.1|$ & 4 & 216 \\
    record & {\lstinline|fn(x:&m Record)|} & {\lstinline|rev2(x.l);min(x.t, 0);|} & $9+18|x.l|+5|x.t|$ & 4 & 216 \\
    sum\_rec\_nlist & {\lstinline|fn(l:&NList) -> i32|} & sum up {\lstinline|NList|} as \lstinline|sum_rec| & $1+5|l|$ & 16 & 26 \\
    rev\_rec\_nlist & {\lstinline|fn(l:&mut NList,r:NList)->NList| } & reverse {\lstinline|NList|} as \lstinline|rev_rec| & $1+7|l|$ & 21 & 89 \\
    \hline
    \end{tabular}
    }
\end{table}
\end{DIFnomarkup}

We used the prototype implementation of \rarust{} to perform an experimental evaluation of resource analysis on typical examples concerning Rust's borrow mechanism. We adapted several pure functional examples from \textsc{RaML} to their Rust counterparts employing borrows; by using borrows we can do in-place updates in Rust. Some examples were deliberately crafted to demonstrate the prototype’s capability to support the aforementioned features. Due to the linear limitation, we select those examples with linear complexity. The experiments were performed on a laptop with 2.20 GHz Intel Core i9-13900HX as CPU and WSL 2.3.24.0/Ubuntu 22.04.3 LTS as operation system. The compiling of the benchmarks was done in 0.15s and the analysis was done in 0.3s.

\cref{tab:eval} shows our experimental results. We manually checked the analysis results on the benchmarks and confirmed that the derived bounds are correct (but not tight for some benchmarks such as \lstinline|min| and \lstinline|max|). We encode some cases with the derived bounds in \textsc{Flux}~\cite{Flux} to check that the derived bounds are correct. Our encoding introduces a global counter to accumulate resource consumption. The encoding is shipped into the artifact and it includes cases \lstinline|append|, \lstinline|concat|, \lstinline|sum_rec|, \lstinline|rev_rec|, \lstinline|sum2|, \lstinline|rev|, \lstinline|rev2|, \lstinline|dup|, \lstinline|dup2|, \lstinline|min|, \lstinline|max|, \lstinline|find|, \lstinline|add|, \lstinline|tuple|, \lstinline|record|, \lstinline|sum_rec_nlist|, and \lstinline|rev_rec_nlist|.

\cref{tab:eval} gives out 5 groups of test cases, and each group exercises some features. For each analyzed function as a case, \cref{tab:eval} first declares the function type in a simplified syntax of Rust, writing \lstinline|&mut T| as \lstinline|&m T| for short. The description column provides a concise explanation of the function's semantics. We plot the complexity in a more readable format in the table, where $|l|$ represents the length or count of \lstinline|Cons| constructors of $l$ and $|t|$ represents the count of \lstinline|Node| constructors of the tree $t$. The concrete coefficients are inferred by \rarust{} according to annotation $\kwd{tick}(\delta)$ in examples. In our evaluation, we set those different concrete numbers of $\delta$ for two purposes: (i) we roughly add one $\kwd{tick}(\delta)$ around one statement to account for the number of operations by the statement, and (ii) we can use different numbers to test multiple times to check if our implementation is correct.
We will explain each group in detail in the rest of this section.

To show that \rarust{} can handle mutable borrows, we construct cases \lstinline|end_m| and \lstinline|end_c|. They are recursively to find the mutable borrow of the nil of a list $l$. For example, in ML syntax, the nil of the list \verb|1::2::3::4::[]| is \verb|[]|. \lstinline|end_m| returns the borrow, while  \lstinline|end_c| storing it in the parameter \lstinline|o:&m&m List|. The returned mutable borrow of \lstinline|end_m| works as a closure function with type \lstinline|List->List|, therefore it is non-trivial for resource analysis. We use cases \lstinline|append| and \lstinline|concat| to show the resource correctness as well as the compositionality of the analysis.

Rust programmers are able to write code with loop statements. We construct cases \lstinline|sum_rec|, \lstinline|sum_loop|, \lstinline|rev_rec| and \lstinline|rev_loop|. We focus on shared borrows in \lstinline|sum| and on mutable borrows in \lstinline|rev|. The suffix \lstinline|rec| means recursive function and \lstinline|loop| means loop statements \lstinline|while true { ... }|. The same analysis results reveal that both are supported by \rarust{}. 

We construct multiple calls of function for shared borrows and mutable borrows, to demonstrate sharing and prophesying. The suffix \lstinline|2| means twice in cases \lstinline|sum2|, \lstinline|rev2| and \lstinline|dup2|. The coefficients in the complexity of \lstinline|sum2| are exactly 2 times of \lstinline|sum|, testing the sharing for shared borrows. \lstinline|rev2| is similar but for the prophesying of mutable borrows. \lstinline|dup2| is made to point out the difference between sharing and prophesying. The function \lstinline|dup| mutates the length of list $l$, therefore the second call of \lstinline|dup| is with length $2|l|$, making the linear coefficient of \lstinline|dup2| be $33$, 3 times $11$.

\rarust{} also supports reborrows and nested borrows. We construct cases \lstinline|reborrow_s|, \lstinline|reborrow_m|, \lstinline|nested_s_s|, \lstinline|nested_m_s| and \lstinline|nested_m_m|. The suffix \lstinline|s| denotes shared borrows, while \lstinline|m| denotes mutable borrows.

Besides \lstinline|List|, we construct simple examples for other safe user-defined data types like trees, tuples, records, and C-style lists. The sizes of trees are the counts of internal nodes \lstinline|Node|, instead of the intuitive measure, their depths. \lstinline|Tuple| and \lstinline|Record| are data types to compose \lstinline|List| and \lstinline|Tree|, with complexity as the linear composition of their fields', such as \lstinline|x.0| and \lstinline|x.t|. \rarust{} also supports mutually recursive data types like \lstinline|NList| and \lstinline|NListNode|; they are C-style lists, with the former as the nullable pointer, the latter as the internal node of lists.
\section{Discussion}
\label{sec:discussion}

In this section, we discuss some limitations of \rarust{} mentioned in \cref{section:introduction} and propose possible pathways towards overcoming them to improve the capability of \rarust{} in future work.

\paragraph{Unsafe code, interior mutability, vectors, reference counting, and cyclic data structures}
We focus on safe Rust programs because our design of \rarust{} relies on guarantees provided by Rust's borrow mechanisms, e.g., aliasing and mutation cannot happen simultaneously.
However, Rust programs cannot avoid unsafe code in general, because many standard libraries---including cells (\verb|Cell|), vectors (\verb|Vec|), and reference counting (\verb|Rc|)---rely on unsafe code to allow shared mutable states, e.g., interior mutability.
The unsafe code can operate C-style pointers in an unrestricted way and compromise Rust's memory safety; as a result, \rarust{}'s resource analysis cannot handle unsafe code.
This is actually a common limitation of advanced type systems and verification frameworks for Rust, including Flux~\cite{Flux}, Aeneas~\cite{Aeneas}, and Prusti~\cite{OOPSLA:AMP19}.
Nevertheless, there have been efforts to support unsafe code in formal reasoning about Rust programs.
RustBelt~\cite{RustBelt} pioneers a line of work on semantic typing and separation-logic-based verification of Rust programs with unsafe code. 
Verus~\cite{OOPSLA:LHC23} supports some unsafe features by providing specifications for unsafe memory operations to be memory safe and employing SMT solvers to check those specifications automatically.
However, it is unclear if one can integrate AARA type systems with those techniques.

One common method to support unsafe code is based on Rust's design philosophy: unsafe operations should be properly \emph{encapsulated} by safe APIs, and the developers of those unsafe operations take charge of ensuring the unsafe code does not break Rust's memory safety.
In terms of type systems, this amounts to assigning types to the safe APIs instead of inferring types from the unsafe code body.
Therefore, it would be possible for \rarust{} to analyze Rust programs with unsafe code, if all the unsafe code is encapsulated by resource-annotated safe APIs.
Rust's \verb|Cell| features interior mutability by providing operations for both getting and setting the content of a memory location.
At the API level, the \verb|Cell| type works similarly to references in an ML-like functional programming language, so it would be possible to adapt an AARA approach for supporting references~\cite{FSCD:LH17}.
For example, the resource-annotated APIs shown below can be used to manipulate cells storing lists:
\begin{align*}
    \texttt{new}: & \kwd{fn}(\texttt{l}: \kwd{list}(\alpha) ) \to \texttt{Cell<}\kwd{list}(\alpha)\texttt{>} | 0, \\
    \texttt{replace} : & \kwd{fn}(\texttt{self}: \&\texttt{Cell<}\kwd{list}(\alpha)\texttt{>}, \texttt{l}: \kwd{list}(\alpha) ) \to \kwd{list}(\alpha) | 0 ,
\end{align*}
in the sense that the potential type $\kwd{list}(\alpha)$ is an invariant for a cell type, and operations should maintain the invariant, e.g., \verb|replace| should store another list of the same type $\kwd{list}(\alpha)$.
Rust's \verb|Vec| also makes use of unsafe code to allow accessing uninitialized memory.
At the API level, we can treat vectors as abstract dynamic arrays, which fit nicely into the AARA framework because of their amortized complexity.
For example, we can declare the following APIs for integer vectors: 
\begin{align*}
    \texttt{new}: & \kwd{fn}() \to \texttt{Vec<}\kwd{i32},\alpha\texttt{>} | 0, \\
    \texttt{push} : & \kwd{fn}(\texttt{self}: \&\kwd{mut}~\texttt{Vec<}\kwd{i32},\alpha\texttt{>}, \texttt{n}: \kwd{i32} ) \to () | \alpha + 4,
\end{align*}
where the resource-annotated type $\texttt{Vec<}\kwd{i32},\alpha\texttt{>}$ denotes the potential function $\mathit{len} \cdot \alpha + (4 \cdot \mathit{len} - 2 \cdot \mathit{cap})$, with $\mathit{len}$ and $\mathit{cap}$ being the length and capacity of the vector, respectively. 
Intuitively, the potential function states that every vector element carries $\alpha$ units of potential and we need to store $(4 \cdot \mathit{len} - 2 \cdot \mathit{cap})$ units of extra potential for vector resizing, which would consume $2 \cdot \mathit{len}$ units of resource to extend the vector's capacity when the vector becomes full.

Rust's implementation of \verb|Rc| uses unsafe code, so we would annotate \verb|Rc| APIs with resource-annotated types. \verb|Rc| itself does not permit mutation, so we could model its behavior as if it is a shared reference: \verb|Rc::new()| should store potentials (e.g., \verb|Rc<list(4)>|) and \verb|Rc::clone()| should split potentials (e.g., splitting \verb|Rc<list(4)>| as \verb|Rc<list(1)>| and \verb|Rc<list(3)>|). We have not yet considered multithreading, and supporting \verb|Arc| would be interesting future work.

It is non-trivial to handle cyclic data structures. Rust provides \verb|Weak| pointers to accompany \verb|Rc| pointers. However, creating cyclic data structures usually requires using interior mutability (e.g., \verb|Cell| or \verb|RefCell|). The interaction between reference counting and interior mutability seems quite non-trivial. In the future, we may adapt \citet{ESOP:Atkey10}'s work on integrating AARA with separation logic.

\paragraph{Generic types, higher-order functions (closures), and trait objects}
Current \rarust{} does not support generic types like \verb|List<T>|. This is not a fundamental limitation, because we can always instantiate generic types. We will spare engineering efforts to support them in the future.

Our work currently only considers top-level functions, but Rust does support higher-order functions and closures to enable functional programming style. 
Fortunately, many AARA approaches support higher-order functions~\cite{AARA-HigherOrder,POPL:HDW17,ICFP:KWR20,ICFP:KH21}.
Conceptually,
it would be possible to adapt AARA's methodology of handling closures to \rarust{}
by extending the type system to deal with \emph{capturing} properly.
One simple extension is to enforce that closures cannot consume potentials stored in captured variables; in this way, it is sound to apply a closure multiple times.
It would be interesting future research to investigate how the interaction of borrow mechanisms (especially mutable borrows) and closures would affect AARA-style resource analysis.

Rust supports a form of dynamic dispatch through trait objects, in which the compiler knows an object's trait but not its actual type. One possible workaround is to annotate trait methods with resource annotations and require them not to change the resource type of \verb|Self|. In this way, even though we do not know an object's type, we know how calling its trait methods affects the resource-annotated context. Another possibility is to adapt \citet{AARA-OOP}'s work on integrating AARA with objective-oriented programming. 

\paragraph{Non-linear resource bounds}
Both our formalization and implementation of \rarust{} currently only consider linear resource bounds, which are too restrictive to analyze real-world programs.
Current \rarust{} can only support pattern matching of one single variable, without primitive support for pattern matching of tuple types, because tuple types usually introduce multivariate polynomial resource bounds like $\textit{first} \times \textit{second}$. 
This is not a fundamental limitation because it has been shown that AARA type systems can support polynomial bounds~\cite{AARA-Poly,AARA-Poly-Multivar}, exponential bounds~\cite{AARA-Exp}, logarithmic bounds~\cite{AARA-Log,CAV:LMZ21}, and value-dependent bounds~\cite{ICFP:KWR20,PLDI:KWP19}.
All of those approaches amount to devising proper type annotations that specify particular kinds of potential functions and developing constraint-based type-inference algorithms.
We plan to spare engineering efforts to extend our prototype implementation of \rarust{} to support various classes of resource bounds in the future.

\section{Related Work} \label{sec:related}

In this section, we discuss related work that has not been covered in previous sections.

\paragraph{Static resource analysis}
AARA is not the only approach for type-based resource analysis.
For example, there are approaches based on sized types~\cite{phd:Vasconcelos08,ICFP:AL17}, dependent types~\cite{LICS:LG11,POPL:LP13,POPL:RGG21}, refinement types~\cite{POPL:HVH20,POPL:CBG17,ESOP:CGA15,OOPSLA:WWC17},
recurrence extraction~\cite{POPL:KML20,ICFP:DLR15}, logical frameworks~\cite{POPL:NSG22,POPL:GNS24}, and annotated types~\cite{POPL:CW00,POPL:Danielsson08}.
None of the aforementioned type systems considered supporting heap-manipulating programs with Rust's borrow mechanisms.
Besides type systems, there are also static resource-analysis techniques based on
defunctionalization~\cite{ICFP:ALM15}, recurrence relations~\cite{JAR:AAG11,TACAS:AFR15,APLAS:FH14,PLDI:BCK20,POPL:KBC19,PLDI:KBB17},
term rewriting~\cite{RTA:AM13,TACAS:BEG14,IJCAR:FNH16,JAR:NEG13}, and
abstract interpretation~\cite{SAS:ZSG11,LPAR:BHH10,CAV:SZV14,kn:DHW07,misc:fbinfer20,SAS:AG12}.
Some of the aforementioned techniques work on imperative programs (e.g., C programs) with arrays,
such as C4B~\cite{PLDI:CHS15}, SPEED~\cite{POPL:GMC09}, COSTA~\cite{JAR:AAG11}, ICRA~\cite{PLDI:KBB17}, etc., but none of them considered exploiting Rust's borrow mechanisms in the design.
It would be interesting future research direction to investigate how different static resource-analysis techniques can benefit from Rust's safety guarantees.


\paragraph{Graded type system}
Graded types \cite{Granule} introduce a graded modality for associating types with elements from a resource algebra. 
A graded type system can account for program variables' exact usages, security levels, and potentials (conceptually). 
Graded types seem to provide a more general mechanism than AARA types to reason about more general resources. 
On the other hand, our work focuses on how Rust's borrow mechanisms can aid resource analysis and chooses AARA as the starting point because AARA admits efficient type inference via linear programming. 

\paragraph{Program verification for Rust}
RustBelt~\cite{RustBelt} pioneers a line of work to use semantic typing and separation logic to verify Rust programs with both safe and unsafe code.
%
%
RustHorn~\cite{RustHorn} uses prophecy variables to model the future values of mutable borrows and proposes an automatic verification algorithm based on constrained Horn clauses.
%
Aeneas~\cite{Aeneas}---which inspires our development of RABC and supports our prototype implementation---uses LLBC to translate Rust programs into equivalent pure functional programs via symbolic execution.
Such pure functional programs can be ported into theorem provers such as Coq and F* to enable verification of functional correctness.
Prusti~\cite{OOPSLA:AMP19} also leverages Rust's advanced type system to devise a modular and automated verification approach.
\citet{master:Engel21} proposed a method to verify user-provided asymptotic resource bounds in Prusti.
In contrast, our work focuses on the automatic inference of concrete resource bounds via a type system.
%

\section{Conclusion} \label{sec:conclusion}

In this paper, we propose \rarust{}, an AARA-style type-based automatic resource analysis of Rust programs.
Our design of \rarust{} exploits the memory-safety guarantees provided by Rust's borrow mechanisms.
Especially, we propose shared potentials and novel prophecy potentials to deal with shared and mutable borrows in a sound and compositional manner.
We formulate a resource-aware dynamic semantics for \rarust{} and use that to prove the soundness of the type system.
Our prototype implementation of \rarust{} automatically infers linear resource bounds for well-typed and well-borrowed Rust programs.
In the future, we plan to extend \rarust{} to support unsafe code, higher-order functions, and non-linear resource bounds.


\section*{Data-Availability Statement}
The source code of the \rarust{} implementation and benchmarks referenced in \cref{sec:impl} are publicly available in the Zenodo \cite{RaRustArtifact}. The artifact contains necessary scripts and step-by-step guides to reproduce the experimental results.

\begin{acks}
We are grateful to Xuanyu Peng for the early discussion and investigation. We would like to thank the anonymous reviewers for their valuable feedback on our paper and the anonymous artifact reviewers for their suggestions for our artifact.
\end{acks}

\bibliographystyle{ACM-Reference-Format}
\bibliography{main,db}

\newpage
\appendix
\section{Judgements}
\centering
\judgement{Expression Evaluation}{$V\vDash e \rightsquigarrow v$}
    \begin{mathpar}
    \inferrule*[Right=\rulename{V-Ev-Int}]
    {~}
    {V\vDash \kwd{n}_\text{i32} \rightsquigarrow \kwd{n}_\text{i32}}
    \and
    \inferrule*[Right=\rulename{V-Ev-Op}]
    {V\vDash e_1 \rightsquigarrow \kwd{n}_1
    \\ V\vDash e_2 \rightsquigarrow \kwd{n}_2}
    {V\vDash e_1~\kwd{op}~e_2 \rightsquigarrow \kwd{n}_1~\kwd{op}~\kwd{n}_2}
    \\
    \inferrule*[Right=\rulename{V-Ev-True}]
    {~}
    {V\vDash \kwd{true} \rightsquigarrow \kwd{true}}
    \and
    \inferrule*[Right=\rulename{V-Ev-False}]
    {~}
    {V\vDash \kwd{false} \rightsquigarrow \kwd{false}}
    \\
    \inferrule*[Right=\rulename{V-Ev-Nil}]
    {~}
    {V\vDash \kwd{nil} \rightsquigarrow \kwd{nil}}
    \and
    \inferrule*[Right=\rulename{V-Ev-Box}]
    {V\vDash e \rightsquigarrow v}
    {V\vDash \kwd{box}(e) \rightsquigarrow \kwd{box}(v)}
    \\
    \inferrule*[Right=\rulename{V-Ev-Copy}]
    {V\vDash p \rightsquigarrow v}
    {V\vDash \kwd{copy}~p \rightsquigarrow v}
    \and
    \inferrule*[Right=\rulename{V-Ev-Move}]
    {V\vDash p \rightsquigarrow v}
    {V\vDash \kwd{move}~p \rightsquigarrow v}
    \and
    \inferrule*[Right=\rulename{V-Ev-Borrow}]
    {V\vDash p \rightsquigarrow v}
    {V\vDash \&^{\kwd{s}/\kwd{m}/\kwd{2}} p \rightsquigarrow \&(p, v)}
\end{mathpar}

\centering
\judgement{Statement Execution}{$V\vDash e \rightsquigarrow^\delta \Dashv V'$}
\begin{mathpar}
    \inferrule*[Right=\rulename{V-Ex-Ret}]
    {~}
    {V\vDash \kwd{return} \rightsquigarrow^0 \Dashv V}
    \and
    \inferrule*[Right=\rulename{V-Ex-Seq}]
    {V\vDash s_1\rightsquigarrow^{\delta_1}\Dashv V'
    \\ V'\vDash s_2\rightsquigarrow^{\delta_2}\Dashv V''}
    {V\vDash s_1; s_2\rightsquigarrow^{\delta_1+\delta_2}\Dashv V''}
    \\
    \inferrule*[Right=\rulename{V-Ex-Tick}]
    {~}
    {V\vDash \kwd{tick}(\delta)\rightsquigarrow^\delta \Dashv V}
    \and
    \inferrule*[Right=\rulename{V-Ex-Drop}]
    {~}
    {V\vDash \kwd{drop}~p\rightsquigarrow^0\Dashv V}
    \\
    \inferrule*[Right=\rulename{V-Ex-Assign}]
    {V\vDash e \rightsquigarrow v'
    \\ \VWt{V}{p}{v'}{V'}}
    {V\vDash p\from e\rightsquigarrow^0 \Dashv V'}
    \\
    \inferrule*[Right=\rulename{V-Ex-Cons}]
    {V\vDash e_1 \rightsquigarrow v_1
    \\ V\vDash e_2 \rightsquigarrow v_2
    \\ \VWt{V}{p}{\kwd{cons}(v_1, v_2)}{V'} }
    {V\vDash p\from \kwd{cons}(e_1, e_2)\rightsquigarrow^0 \Dashv V'}
    \\
    \inferrule*[Right=\rulename{V-Ex-IfT}]
    {V\vDash p\rightsquigarrow \kwd{true}
    \\ V\vDash s_1\rightsquigarrow^\delta \Dashv V'}
    {V\vDash \kwd{if}~ p ~\kwd{then}~ s_1 ~\kwd{else}~ s_2 ~\kwd{end} \rightsquigarrow^\delta \Dashv V'}
    \and
    \inferrule*[Right=\rulename{V-Ex-IfF}]
    {V\vDash p\rightsquigarrow \kwd{false}
    \\ V\vDash s_2\rightsquigarrow^\delta \Dashv V'}
    {V\vDash \kwd{if}~ p ~\kwd{then}~ s_1 ~\kwd{else}~ s_2 ~\kwd{end} \rightsquigarrow^\delta \Dashv V'}
    \\
    \inferrule*[Right=\rulename{V-Ex-Mat-Nil}]
    {V\vDash p\rightsquigarrow \kwd{nil}
    \\ V\vDash s_1\rightsquigarrow^\delta \Dashv V'}
    {V\vDash \kwd{match}~ p ~ \{\kwd{nil}\mapsto s_1, \kwd{cons}(x_\text{hd}, x_\text{tl})\mapsto s_2\} \rightsquigarrow^\delta \Dashv V'}

    \inferrule*[Right=\rulename{V-Ex-Mat-Cons}]
    {V\vDash p\rightsquigarrow \kwd{cons}(hd, tl)
    \\ \VWt{V}{p}{\bot}{V_1}
    \\ \VWt{V_1}{x_\text{hd}}{hd}{V_2}
    \\ \VWt{V_2}{x_\text{tl}}{tl}{V_\text{b}}
    \\\\ V_\text{b}\vDash s_2\rightsquigarrow^\delta \Dashv V'_\text{b}
    \\ V'_\text{b}\vDash x_\text{hd}\rightsquigarrow hd'
    \\ V'_\text{b}\vDash x_\text{tl}\rightsquigarrow tl'
    \\\\ \VWt{V'_\text{b}}{x_\text{hd}}{\bot}{V'_1}
    \\ \VWt{V'_1}{x_\text{tl}}{\bot}{V'_2}
    \\ \VWt{V'_2}{p}{\kwd{cons}(hd', tl')}{V'} 
    }
    {V\vDash \kwd{match}~ p ~ \{\kwd{nil}\mapsto s_1, \kwd{cons}(x_\text{hd}, x_\text{tl})\mapsto s_2\} \rightsquigarrow^\delta \Dashv V'}
    \\
    
    \inferrule*[Right=\rulename{V-Ex-App}]
    {\kwd{fn}~ f ~(\vec{x}_\text{param}:\vec{t}_\text{param}, \vec{x}_\text{local}:\vec{t}_\text{local}, x_\text{ret}:t_\text{ret}) \{~ s ~\}
    \\ V\vDash \vec{e}\rightsquigarrow \vec{v}
    \\\\ \VWt{V}{\vec{x}_\text{param}}{\vec{v}}{V_1}
    \\ \VWt{V_1}{\vec{x}_\text{local}}{\bot}{V_2}
    \\ \VWt{V_2}{x_\text{ret}}{\bot}{V_\text{b}}
    \\\\ V_\text{b}\vDash s\rightsquigarrow^\delta \Dashv V'_\text{b}
    \\ V'_\text{b}\vDash x_\text{ret} \rightsquigarrow v_\text{ret}
    \\\\ \VWt{V'_\text{b}}{\vec{x}_\text{param}}{\bot}{V'_1}
    \\ \VWt{V'_1}{\vec{x}_\text{local}}{\bot}{V'_2}
    \\ \VWt{V'_2}{x_\text{ret}}{\bot}{V'_3}
    \\ \VWt{V'_3}{p}{v_\text{ret}}{V'}
    } 
    {V\vDash p\from f(\vec{e})\rightsquigarrow^\delta \Dashv V'}
\end{mathpar}

\centering
\judgement{Enrich}{$\textit{enrich}~ t ~\textit{as}~ \tau$}
\begin{mathpar}
    \inferrule*[Right=\rulename{Enrich-Int}]
    {~}
    {\textit{enrich}~ \kwd{i32} ~\textit{as}~\kwd{i32}}
    \and
    \inferrule*[Right=\rulename{Enrich-Bool}]
    {~}
    {\textit{enrich}~ \kwd{bool} ~\textit{as}~ \kwd{bool}}
    \\

    \inferrule*[Right=\rulename{Enrich-List}]
    {\alpha~\text{fresh}}
    {\textit{enrich}~ \kwd{list} ~\textit{as}~  \kwd{list}(\alpha)}
    \and
    \inferrule*[Right=\rulename{Enrich-Box}]
    {\alpha~\text{fresh}}
    {\textit{enrich}~ \kwd{box}(\kwd{list}) ~\textit{as}~ \kwd{box}(\kwd{list}(\alpha))}
    \\
    
    \inferrule*[Right=\rulename{Enrich-Shared}]
    {\textit{enrich}~ t ~\textit{as}~ \tau}
    {\textit{enrich}~ \&^\kwd{s}(t) ~\textit{as}~ \&^\kwd{s}(\tau)}
    \\
    \inferrule*[Right=\rulename{Enrich-Mutable}]
    {\textit{enrich}~ t ~\textit{as}~ \tau_\text{c}
    \\ \textit{enrich}~ t ~\textit{as}~ \tau_\text{p}
    }
    {\textit{enrich}~ \&^\kwd{m}(t) ~\textit{as}~ \&^\kwd{m}(\tau_\text{c}, \tau_\text{p})}
\end{mathpar}

\centering
\judgement{Sharing}{$\textit{share}~ \tau ~\textit{as}~\tau_1, \tau_2$}
\begin{mathpar}
    \inferrule*[Right=\rulename{Share-Int}]
    {~}
    {\textit{share}~ \kwd{i32} ~\textit{as}~\kwd{i32}, \kwd{i32}}
    \and
    \inferrule*[Right=\rulename{Share-Bool}]
    {~}
    {\textit{share}~ \kwd{bool} ~\textit{as}~\kwd{bool}, \kwd{bool}}
    \\
    \inferrule*[Right=\rulename{Share-List}]
    {\alpha_1, \alpha_2 ~\text{fresh}
    \\\alpha = \alpha_1 + \alpha_2}
    {\textit{share}~ \kwd{list}(\alpha) ~\textit{as}~\kwd{list}(\alpha_1), \kwd{list}(\alpha_2)}
    \\
    \inferrule*[Right=\rulename{Share-Box}]
    {\textit{share}~ \tau ~\textit{as}~\tau_1, \tau_2}
    {\textit{share}~ \kwd{box}(\tau) ~\textit{as}~\kwd{box}(\tau_1), \kwd{box}(\tau_2)}
    \\
    \inferrule*[Right=\rulename{Share-Shared}]
    {\textit{share}~ \tau ~\textit{as}~\tau_1, \tau_2}
    {\textit{share}~ \&^\kwd{s}(\tau) ~\textit{as}~ \&^\kwd{s}(\tau_1), \&^\kwd{s}(\tau_2)}
\end{mathpar}

\centering
\judgement{Prophesying}{$\textit{prophesy}~ \tau_\text{c} ~\textit{as}~\tau_\text{p}$}
\begin{mathpar}
    \inferrule*[Right=\rulename{Prophesy-Int}]
    {~}
    {\textit{prophesy}~ \kwd{i32} ~\textit{as}~ \kwd{i32}}
    \and
    \inferrule*[Right=\rulename{Prophesy-Bool}]
    {~}
    {\textit{prophesy}~ \kwd{bool} ~\textit{as}~ \kwd{bool}}
    \\
    \inferrule*[Right=\rulename{Prophesy-List}]
    {\alpha_\text{p}~\text{fresh}}
    {\textit{prophesy}~ \kwd{list}(\alpha) ~\textit{as}~ \kwd{list}(\alpha_\text{p})}
    \\
    \inferrule*[Right=\rulename{Prophesy-Box}]
    {\textit{prophesy}~ \tau ~\textit{as}~ \tau_\text{p} }
    {\textit{prophesy}~ \kwd{box}(\tau) ~\textit{as}~ \kwd{box}(\tau_\text{p})}
    \\
    \inferrule*[Right=\rulename{Prophesy-Shared}]
    {\textit{prophesy}~ \tau ~\textit{as}~ \tau_\text{p}}
    {\textit{prophesy}~ \&^\kwd{s}(\tau) ~\textit{as}~ \&^\kwd{s}(\tau_\text{p})}
    \\
    \inferrule*[Right=\rulename{Prophesy-Mutable}]
    {\textit{prophesy}~ \tau_\text{c} ~\textit{as}~ \tau_\text{cp}
    \\ \textit{prophesy}~ \tau_\text{p} ~\textit{as}~ \tau_\text{pp}
    }
    {\textit{prophesy}~ \&^\kwd{m}(\tau_\text{c}, \tau_\text{p}) ~\textit{as}~ \&^\kwd{m}(\tau_\text{cp}, \tau_\text{pp})}
\end{mathpar}

\centering
\judgement{Meet/Join}{$\tau_1\cap\tau_2 / \tau_1\cup\tau_2$}
\begin{mathpar}    
    \inferrule*[Right=Meet-Int]
    {~}
    {\kwd{i32}\cap\kwd{i32}=\kwd{i32}}
    \and
    \inferrule*[Right=Join-Int]
    {~}
    {\kwd{i32}\cup\kwd{i32}=\kwd{i32}}
    \\
    \inferrule*[Right=Meet-Bool]
    {~}
    {\kwd{bool}\cap\kwd{bool}=\kwd{bool}}
    \and
    \inferrule*[Right=Join-Bool]
    {~}
    {\kwd{bool}\cup\kwd{bool}=\kwd{bool}}
    \\
    
    \inferrule*[Right=Meet-List]
    {\min(\alpha_1, \alpha_2)=\alpha}
    {\kwd{list}(\alpha_1)\cap\kwd{list}(\alpha_2)=\kwd{list}(\alpha)}
    \and
    \inferrule*[Right=Join-List]
    {\max(\alpha_1, \alpha_2)=\alpha}
    {\kwd{list}(\alpha_1)\cup\kwd{list}(\alpha_2)=\kwd{list}(\alpha)}
    \\
    
    \inferrule*[Right=Meet-Box]
    {\tau_1 \cap \tau_2=\tau}
    {\kwd{box}(\tau_1)\cap\kwd{box}(\tau_2)=\kwd{box}(\tau)}
    \and
    \inferrule*[Right=Join-Box]
    {\tau_1 \cup \tau_2=\tau}
    {\kwd{box}(\tau_1)\cup\kwd{box}(\tau_2)=\kwd{box}(\tau)}
    \\
    
    \inferrule*[Right=Meet-Shared]
    {\tau_1 \cap \tau_2=\tau}
    {\&^\kwd{s}(\tau_1)\cap\&^\kwd{s}(\tau_2)=\&^\kwd{s}(\tau)}
    \and
    \inferrule*[Right=Join-Shared]
    {\tau_1 \cup \tau_2=\tau}
    {\&^\kwd{s}(\tau_1)\cup\&^\kwd{s}(\tau_2)=\&^\kwd{s}(\tau)}
    \\

    \inferrule*[Right=Meet-Mutable]
    {\tau_{\text{c}, 1} \cap \tau_{\text{c}, 2}=\tau_\text{c}
    \\ \tau_{\text{p}, 1} \cup \tau_{\text{p}, 2}=\tau_\text{p}
    \\ \vdash \&^\kwd{m}(\tau_{\text{c}, 1}, \tau_{\text{p}, 1})
    \\ \vdash \&^\kwd{m}(\tau_{\text{c}, 2}, \tau_{\text{p}, 2})
    }
    {\&^\kwd{m}(\tau_{\text{c}, 1}, \tau_{\text{p}, 1})\cap\&^\kwd{m}(\tau_{\text{c}, 2}, \tau_{\text{p}, 2})=\&^\kwd{m}(\tau_\text{c}, \tau_\text{p})}
    \\
    \inferrule*[Right=Join-Mutable]
    {\tau_{\text{c}, 1} \cup \tau_{\text{c}, 2}=\tau_\text{c}
    \\ \tau_{\text{p}, 1} \cap \tau_{\text{p}, 2}=\tau_\text{p}
    \\ \vdash \&^\kwd{m}(\tau_{\text{c}, 1}, \tau_{\text{p}, 1})
    \\ \vdash \&^\kwd{m}(\tau_{\text{c}, 2}, \tau_{\text{p}, 2})
    }
    {\&^\kwd{m}(\tau_{\text{c}, 1}, \tau_{\text{p}, 1})\cup\&^\kwd{m}(\tau_{\text{c}, 2}, \tau_{\text{p}, 2})=\&^\kwd{m}(\tau_\text{c}, \tau_\text{p})}
\end{mathpar}

\centering
\judgement{Typing Expression}{$\Gamma\vdash e \hookrightarrow \tau\dashv\Gamma'$}
\begin{mathpar}
    \inferrule*[Right=\rulename{$\Gamma$-Ev-Int}]
    {~}
    {\Gamma\vdash \kwd{n}_\text{i32} \hookrightarrow \kwd{i32}\dashv\Gamma}
    \and
    \inferrule*[Right=\rulename{$\Gamma$-Ev-Op}]
    {\Gamma\vdash e_1\hookrightarrow \kwd{i32}\dashv\Gamma
    \\ \Gamma\vdash e_2\hookrightarrow \kwd{i32}\dashv\Gamma}
    {\Gamma\vdash e_1~\kwd{op}~e_2 \hookrightarrow \kwd{i32}\dashv\Gamma}
    \\
    
    \inferrule*[Right=\rulename{$\Gamma$-Ev-True}]
    {~}
    {\Gamma\vdash \kwd{true} \hookrightarrow \kwd{bool}\dashv\Gamma}
    \and
    \inferrule*[Right=\rulename{$\Gamma$-Ev-False}]
    {~}
    {\Gamma\vdash \kwd{false} \hookrightarrow \kwd{bool}\dashv\Gamma}
    \\
    \inferrule*[Right=\rulename{$\Gamma$-Ev-Copy}]
    {\Gamma\vdash p \hookrightarrow \tau
    \\ \tau = \kwd{i32} ~\text{or}~ \tau = \kwd{bool} }
    {\Gamma\vdash \kwd{copy}~p \hookrightarrow \tau\dashv\Gamma}
    \and
    \inferrule*[Right=\rulename{$\Gamma$-Ev-Box}]
    {\Gamma\vdash e \hookrightarrow \tau\vdash\Gamma'}
    {\Gamma\vdash \kwd{box}(e) \hookrightarrow \kwd{box}(\tau)\vdash\Gamma'}
    \\

    \inferrule*[Right=\rulename{$\Gamma$-Ev-Nil}]
    {\alpha ~\text{fresh}}
    {\Gamma\vdash \kwd{nil} \hookrightarrow \kwd{list}(\alpha)\vdash\Gamma}
    \and
    \inferrule*[Right=\rulename{$\Gamma$-Ev-Move}]
    {\Gamma\vdash p \hookrightarrow \tau
    \\ \GWt{\Gamma}{p}{\bot}{\Gamma'}}
    {\Gamma\vdash \kwd{move}~p \hookrightarrow \tau\dashv\Gamma'}
    \\
    
    \inferrule*[Right=\rulename{$\Gamma$-Ev-Shared}]
    {\Gamma\vdash p \hookrightarrow \tau
    \\ \textit{share}~ \tau ~\textit{as}~\tau_1, \tau_2
    \\ \GWt{\Gamma}{p}{\tau_1}{\Gamma'}
    }
    {\Gamma\vdash \&^\kwd{s}~p \hookrightarrow \&^\kwd{s}(\tau_2)\dashv\Gamma'}
    \\
    
    \inferrule*[Right=\rulename{$\Gamma$-Ev-Mutable}]
    {\Gamma\vdash p \hookrightarrow \tau
    \\ \textit{prophesy}~ \tau ~\textit{as}~ \tau_\text{p} 
    \\ \GWt{\Gamma}{p}{\tau_\text{p}}{\Gamma'}
    }
    {\Gamma\vdash \&^\kwd{m}~p \hookrightarrow \&^\kwd{m}(\tau, \tau_\text{p})\dashv\Gamma'}
\end{mathpar}

\centering
\judgement{Typing Statements}{$\Gamma\vdash s \hookrightarrow^\delta \dashv\Gamma'$}
\begin{mathpar}
    \inferrule*[Right=\rulename{$\Gamma$-Ex-Ret}]
    {~}
    {\Gamma\vdash \kwd{return} \hookrightarrow^0\vdash\Gamma}
    \and
    \inferrule*[Right=\rulename{$\Gamma$-Ex-Seq}]
    {\Gamma_1\vdash s_1\hookrightarrow^{\delta_1}\dashv\Gamma_2
    \\ \Gamma_2\vdash s_2\hookrightarrow^{\delta_2}\dashv\Gamma_3}
    {\Gamma_1\vdash s_1; s_2\hookrightarrow^{\delta_1+\delta_2}\dashv\Gamma_3}
    \\
    
    \inferrule*[Right=\rulename{$\Gamma$-Ex-Tick}]
    {~}
    {\Gamma\vdash\kwd{tick}(\delta)\hookrightarrow^\delta\vdash\Gamma}
    \and
    \inferrule*[Right=\rulename{$\Gamma$-Ex-Drop}]
    {\Gamma\vdash p\hookrightarrow \tau
    \\ \vdash \tau
    \\ \GWt{\Gamma}{p}{\bot}{\Gamma'}
    }
    {\Gamma\vdash \kwd{drop}~p \hookrightarrow^0\dashv \Gamma'}
    \\

    \inferrule*[Right=\rulename{$\Gamma$-Ex-Assign}]
    {\Gamma\vdash e\hookrightarrow \tau'\dashv\Gamma_1
    \\ \Gamma_1\vdash p \hookrightarrow \tau
    \\ \vdash \tau
    \\ \GWt{\Gamma_1}{p}{\tau'}{\Gamma'}}
    {\Gamma\vdash p\from e \hookrightarrow^0\dashv\Gamma'}
    \\

    \inferrule*[Right=\rulename{$\Gamma$-Ex-Cons}]
    {\Gamma\vdash e_1\hookrightarrow \kwd{i32} \dashv \Gamma_1
    \\ \Gamma_1\vdash e_2\hookrightarrow \kwd{box}(\kwd{list}(\alpha'))\dashv\Gamma_2
    \\ \GWt{\Gamma_2}{p}{\kwd{list}(\alpha')}{\Gamma'}}
    {\Gamma\vdash p\from \kwd{cons}(e_1, e_2)\hookrightarrow^{\alpha'}\dashv\Gamma'}
    \\

    \inferrule*[Right=\rulename{$\Gamma$-Ex-If}]
    {\Gamma\vdash p\hookrightarrow \kwd{bool}
    \\ \Gamma\vdash s_1\hookrightarrow^{\delta_1}\dashv\Gamma_1
    \\ \Gamma\vdash s_2\hookrightarrow^{\delta_2}\dashv\Gamma_2
    \\ \max(\delta_1, \delta_2)=\delta
    \\ \Gamma_1\sqcap\Gamma_2=\Gamma' }
    {\Gamma\vdash \kwd{if}~ p ~\kwd{then}~ s_1 ~\kwd{else}~ s_2 ~\kwd{end} \hookrightarrow^\delta \dashv\Gamma'}
    \\
    \inferrule*[Right=\rulename{$\Gamma$-Ex-Mat}]
    {\Gamma\vdash p\hookrightarrow \kwd{list}(\alpha)
    \\ \Gamma\vdash s_1\hookrightarrow^{\delta_1}\dashv\Gamma_1
    \\\\ \GWt{\Gamma}{p}{\bot}{\Gamma_{\text{b}, 1}}
    \\ \GWt{\Gamma_{\text{b}, 1}}{x_\text{hd}}{\kwd{i32}}{\Gamma_{\text{b}, 2}}
    \\ \GWt{\Gamma_{\text{b}, 2}}{x_\text{tl}}{\kwd{box}(\kwd{list}(\alpha))}{\Gamma_\text{b}}
    \\ \Gamma_\text{b}\vdash s_2\hookrightarrow^{\delta_2}\dashv\Gamma'_\text{b}
    \\\\ \Gamma'_\text{b}\vdash x_\text{tl}\hookrightarrow \kwd{list}(\beta)
    \\ \GWt{\Gamma'_\text{b}}{x_\text{hd}}{\bot}{\Gamma'_{\text{b}, 1}}
    \\ \GWt{\Gamma'_{\text{b}, 1}}{x_\text{tl}}{\bot}{\Gamma'_{\text{b}, 2}}
    \\ \GWt{\Gamma'_{\text{b}, 2}}{p}{\kwd{list}(\beta)}{\Gamma_2}
    \\\\ \max(\delta_1, \delta_2-(\alpha-\beta))=\delta
    \\ \Gamma_1\sqcap\Gamma_2=\Gamma'}
    {\Gamma\vdash \kwd{match}~ p ~ \{\kwd{nil}\mapsto s_1, \kwd{cons}(x_\text{hd}, x_\text{tl})\mapsto s_2\} \hookrightarrow^\delta \dashv\Gamma'}
    \\

    \inferrule*[Right=\rulename{$\Gamma$-Ex-App}]
    {\text{fn}~ f ~(\vec{x}_\text{param}:\vec{t}_\text{param}, \vec{x}_\text{local}:\vec{t}_\text{local}, x_\text{ret}:t_\text{ret}) \{~ s ~\}
    \\\\ \vdash f \Leftarrow (\Gamma_f, \delta_f)
    \\ \Gamma_f\vdash x_\text{ret} \hookrightarrow \tau_\text{ret}, (\forall x_i\in\vec{x}_\text{param}, i=1, ..., n) \Gamma_f \vdash x_i \hookrightarrow \tau_{\text{param}, i}
    \\\\ \Gamma_0=\Gamma, (\forall e_i\in \vec{e}, i=1, ..., n) \Gamma_{i-1}\vdash e_i\hookrightarrow\tau_{\text{arg}, i}\dashv\Gamma_i
    \\ (\forall i=1,..,n)~ \tau_{\text{param}, i} = \tau_{\text{arg}, i}
    \\ \Gamma_n \vdash p \hookrightarrow \tau
    \\ \vdash \tau
    \\ \GWt{\Gamma_n}{p}{\tau_\text{ret}}{\Gamma'}
    }
    {\Gamma\vdash p\from f(\vec{e})\hookrightarrow^{\delta_f}\dashv\Gamma'}
\end{mathpar}

\centering
\judgement{Function Analysis}{$\vdash f \Leftarrow (\Gamma_f, \delta_f)$}
\begin{mathpar}
    \inferrule
    {\text{fn}~ f ~(\vec{x}_\text{param}:\vec{t}_\text{param}, \vec{x}_\text{local}:\vec{t}_\text{local}, x_\text{ret}:t_\text{ret}) \{~ s ~\}
    \\  \vdash f \Rightarrow (\Gamma_f, \delta_f)
    \\ \Gamma_f\vdash s\hookrightarrow^\delta\dashv\Gamma'_f
    \\\\ \forall x \in \textbf{dom}(\Gamma'_f), \vdash \Gamma'_f(x)
    \\ \Gamma'_f \vdash x_\text{ret} \hookrightarrow \tau'_\text{ret}
    \\ \Gamma_f \vdash x_\text{ret} \hookrightarrow \tau_\text{ret}
    \\ \tau'_\text{ret} = \tau_\text{ret}
    \\ \delta = \delta_f}
    {\vdash f \Leftarrow (\Gamma_f, \delta_f)}
\end{mathpar}

\centering
\judgement{Potential Function}{$\phi(v:\tau)$}
\begin{mathpar}
    \inferrule*[Right=$\phi$-Bot]
    {~}
    {\phi(\_:\bot)=0}
    \and
    \inferrule*[Right=$\phi$-Int]
    {~}
    {\phi(\kwd{n}_\text{i32}:\kwd{i32})=0}
    \\
    \inferrule*[Right=$\phi$-True]
    {~}
    {\phi(\kwd{true}:\kwd{bool})=0}
    \and
    \inferrule*[Right=$\phi$-False]
    {~}
    {\phi(\kwd{false}:\kwd{bool})=0}
    \\
    
    \inferrule*[Right=$\phi$-Nil]
    {~}
    {\phi(\kwd{nil}:\kwd{list}(\alpha))=0}
    \\
    \inferrule*[Right=$\phi$-Cons]
    {~}
    {\phi(\kwd{cons}(\kwd{n}_\text{i32}, \kwd{box}(lv)):\kwd{list}(\alpha))=\alpha+\phi(\kwd{box}(lv):\kwd{box}(\kwd{list}(\alpha)))}
    \\
    \inferrule*[Right=$\phi$-Box]
    {~}
    {\phi(\kwd{box}(lv):\kwd{box}(\kwd{list}(\alpha)))=\phi(lv:\kwd{list}(\alpha))}
    \\
    \inferrule*[Right=$\phi$-Shared]
    {~}
    {\phi(\&(\_, v):\&^\kwd{s}(\tau))=\phi(v:\tau)}
    \\
    \inferrule*[Right=$\phi$-Mutable]
    {~}
    {\phi(\&(\_, v):\&^\kwd{m}(\tau_\text{c}, \tau_\text{p}))=\phi(v:\tau_\text{c})-\phi(v:\tau_\text{p})}
\end{mathpar}
\newpage
\section{Proof of Soundness}
\label{sec:proof}

\begin{lemma}
    Potential is non-negative and keeps subtyping:
    \begin{mathpar}
    \inferrule
    {  \vdash \tau_1
    \\ \vdash \tau_2
    \\ \tau_1 \preceq \tau_2
    }
    { 0 \leq \phi(v:\tau_1) \leq \phi(v:\tau_2)}
    \end{mathpar}
\end{lemma}
\begin{proof}
    First define the size of rich types structural inductively:
    \begin{align*}
    \textbf{size} &: \textbf{RichType} \to \NN \\
    &|~ \bot ~|~ \kwd{i32} ~|~ \kwd{bool} ~|~ \kwd{list}(\_) ~|~ \kwd{box}(\kwd{list}(\_))\mapsto 0 \\
    &|~ \&^\kwd{s}(\tau) \mapsto \textbf{size}(\tau) + 1 \\
    &|~ \&^\kwd{m}(\tau_\text{c}, \tau_\text{p}) \mapsto \textbf{size}(\tau_\text{c}) + \textbf{size}(\tau_\text{p}) + 1
    \end{align*}
    Prove by induction on $\textbf{size}(\tau_1) + \textbf{size}(\tau_2)$, and case on $\tau_1\preceq\tau_2$:
    \begin{enumerate}
        \item {\rulename{S-Bot} $\tau_1=\bot$: \textit{exfalso} due to no derivation for $\vdash \bot$; }
        \item {\rulename{S-Int}, \rulename{S-Bool} : $0 = \phi(v:\tau_1) = \phi(v:\tau_2)$ by definition;}
        \item {\rulename{S-List}, \rulename{S-Box} : $\phi(v:\tau_1) = \alpha_1\cdot n \leq \alpha_2\cdot n$, where $\alpha_1 \leq \alpha_2$ from $\tau_1\preceq \tau_2$ and $n$ is length of $v$;}
        \item {\rulename{S-Shared} : $\tau_1=\&^\kwd{s}(\tau'_1), \tau_2=\&^\kwd{s}(\tau'_2)$ by induction hypothesis on $\tau'_1, \tau'_2$, because $\textbf{size}(\tau'_1)+\textbf{size}(\tau'_2) < \textbf{size}(\tau_1) + \textbf{size}(\tau_2) = \textbf{size}(\tau'_1) + \textbf{size}(\tau'_2) + 2$, $\tau'_1\preceq\tau'_2$ from $\tau_1\preceq\tau_2$ and $\vdash \tau'_1, \vdash \tau'_2$ from $\vdash \tau_1, \vdash \tau_2$;}
        \item {\rulename{S-Mutable} : $\tau_1=\&^\kwd{m}(\tau_{\text{c}, 1}, \tau_{\text{p}, 1}), \tau_2=\&^\kwd{m}(\tau_{\text{c}, 2}, \tau_{\text{p}, 2})$, $v=\&(\_, v')$
        \begin{enumerate}
            \item {$0\leq\phi(v:\tau_1) = \phi(v':\tau_{\text{c}, 1}) - \phi(v':\tau_{\text{p}, 1})$ by induction hypothesis on $\tau_{\text{p}, 1}, \tau_{\text{c}, 1}$, because $\textbf{size}(\tau_{\text{p}, 1})+\textbf{size}(\tau_{\text{c}, 1}) < \textbf{size}(\tau_1)+\textbf{size}(\tau_2)$ and $\tau_{\text{p}, 1}\preceq \tau_{\text{c}, 1}, \vdash \tau_{\text{c}, 1}, \vdash \tau_{\text{p}, 1}$ from $\vdash \tau_1$;}
            \item {$\phi(v:\tau_1)\leq\phi(v:\tau_2)$, i.e. $\phi(v:\tau_{\text{c}, 1})-\phi(v:\tau_{\text{p}, 1})\leq \phi(v:\tau_{\text{c}, 2})-\phi(v:\tau_{\text{p}, 2})$ if and only if $\phi(v:\tau_{\text{c}, 1})\leq \phi(v:\tau_{\text{c}, 2}), \phi(v:\tau_{\text{p}, 2})\leq\phi(v:\tau_{\text{p}, 1})$. The goal can be proved by induction hypothesis because size reducing, subtyping and well-formedness inherited:
            \begin{itemize}
                \item {let $\text{x} = \text{c}$ or $\text{p}$, then $\textbf{size}(\tau_{\text{x}, 1}) + \textbf{size}(\tau_{\text{x}, 2}) < \textbf{size}(\tau_1) + \textbf{size}(\tau_2)$,\\
                the latter $= \textbf{size}(\tau_{\text{c}, 1}) + \textbf{size}(\tau_{\text{p}, 1}) + \textbf{size}(\tau_{\text{c}, 2}) + \textbf{size}(\tau_{\text{p}, 2}) + 2$; }
                \item {$\tau_{\text{c}, 1}\preceq\tau_{\text{c}, 2}, \tau_{\text{p}, 2}\preceq \tau_{\text{p}, 1}$ from $\tau_1\preceq\tau_2$;}
                \item {$\vdash \tau_{\text{c}, 1}, \vdash \tau_{\text{c}, 2}, \vdash \tau_{\text{p}, 1}, \vdash \tau_{\text{p}, 2}$ from $\vdash \tau_1, \vdash \tau_2$.}
            \end{itemize}
            }
        \end{enumerate}
        }
    \end{enumerate}
\end{proof}

\begin{lemma}[Update]
If store or context is written with new value or new type, the difference of potential over store or context is equal to that over new value or new type.
\begin{enumerate}
    \item {If $V\vDash p\rightsquigarrow v$, $\Gamma\vdash p\hookrightarrow \tau$ and $\VWt{V}{p}{v'}{V'}$,
    then $\Phi(V':\Gamma) - \Phi(V:\Gamma)=\Phi(v':\tau)-\Phi(v:\tau)$;}
    \item {If $V\vDash p\rightsquigarrow v$, $\Gamma\vdash p\hookrightarrow \tau$ and $\GWt{\Gamma}{p}{\tau'}{\Gamma'}$, 
    then $\Phi(V:\Gamma')-\Phi(V:\Gamma)=\Phi(v:\tau')-\Phi(v:\tau)$.}
\end{enumerate}
\end{lemma}
\begin{proof}
We first prove the update lemma on $V$, and then on $\Gamma$. \\
By induction on $\VWt{V}{p}{v'}{V'}$:
\begin{enumerate}
    \item {\rulename{V-Wt-Var} : We know from premise that $p = x$, $\forall y \neq x, V'(y) = V(y)$, $V'(x) = v', V(x) = v, \Gamma(x) = \tau$, then we reach
    \begin{align*}
        & \Phi(V':\Gamma)-\Phi(V:\Gamma) \\
        & = \left[\phi(V'(x):\Gamma(x)) + \sum_{y\neq x}\phi(V'(y):\Gamma(y))\right] - \left[\phi(V(x):\Gamma(x)) + \sum_{y\neq x}\phi(V(y):\Gamma(y))\right] \\
        & = \phi(v':\tau) -\phi(v:\tau)
    \end{align*}
    }
    \item {\rulename{V-Wt-Box} : We know from premise that $p = * p_1, V\vDash p_1\rightsquigarrow\kwd{box}(v), \VWt{V}{p_1}{\kwd{box}(v')}{V'}, \Gamma\vdash p_1 \hookrightarrow \kwd{box}(\tau)$, then we reach from hypothesis 
    \begin{align*}
        & \Phi(V':\Gamma)-\Phi(V:\Gamma) \\
        & = \phi(\kwd{box}(v'):\kwd{box}(\tau))-\phi(\kwd{box}(v):\kwd{box}(\tau)) \\
        & = \phi(v':\tau) - \phi(v:\tau)
    \end{align*} 
    }
    \item {\rulename{V-Wt-Borrow} : We know from premise that $p = * p_1, V\vDash p_1\rightsquigarrow\&(q, v), \VWt{V}{q}{v'}{V'}, \VWt{V'}{p_1}{\&(q, v')}{V''}$, and $p_1, q$ are separate and will not form a circle. Also we know $\Gamma\vdash p_1 \hookrightarrow \&^\kwd{m}(\tau, \tau_\text{p})$ \textbf{because only mutable borrows can be updated with values}, $\Gamma\vdash q\hookrightarrow \tau_\text{p}$ because $\Gamma$ is well formed by borrow checker. Therefore by induction hypothesis, we reach 
    \begin{align*}
    &\Phi(V'':\Gamma)-\Phi(V:\Gamma) \\
    &=\Phi(V'':\Gamma)-\Phi(V':\Gamma)+\Phi(V':\Gamma)-\Phi(V:\Gamma) \\
    &= \phi(v':\tau_\text{p})-\phi(v:\tau_\text{p}) + \phi(\&(q, v'):\&^\kwd{m}(\tau, \tau_\text{p}))-\phi(\&(q, v):\&^\kwd{m}(\tau, \tau_\text{p})) \\
    &= \phi(v':\tau_\text{p})-\phi(v:\tau_\text{p}) + [ \phi(v':\tau)-\phi(v':\tau_\text{p})]-[\phi(v:\tau)-\phi(v:\tau_\text{p})] \\
    &= \phi(v':\tau)-\phi(v:\tau)
    \end{align*}
    }
\end{enumerate}
By induction on $\GWt{\Gamma}{p}{\tau'}{\Gamma'}$:
\begin{enumerate}
    \item {\rulename{$\Gamma$-Wt-Var} : We know from premise that $p = x$, $\forall y \neq x, \Gamma'(y) = \Gamma(y)$, $\Gamma'(x) = \tau', \Gamma(x) = \tau, V(x) = v$, then we reach 
    \begin{align*}
        & \Phi(V:\Gamma')-\Phi(V:\Gamma) \\
        & = [\phi(V(x):\Gamma'(x)) + \sum_{y\neq x}\phi(V(y):\Gamma'(y))] - [\phi(V(x):\Gamma(x)) + \sum_{y\neq x}\phi(V(y):\Gamma(y))]\\
        & = \phi(v:\tau') -\phi(v:\tau)
    \end{align*}
    }
    \item {\rulename{$\Gamma$-Wt-Box} : We know from premise that $p = * p_1$, $\Gamma\vDash p_1\rightsquigarrow\kwd{box}(\tau)$, $\GWt{\Gamma}{p_1}{\kwd{box}(\tau')}{\Gamma'}$, and $V\vDash p_1 \rightsquigarrow \kwd{box}(v)$, then we reach from hypothesis 
    \begin{align*}
        & \Phi(V:\Gamma')-\Phi(V:\Gamma) \\
        & = \phi(\kwd{box}(v):\kwd{box}(\tau'))-\phi(\kwd{box}(v):\kwd{box}(\tau)) \\
        & = \phi(v:\tau') - \phi(v:\tau)
    \end{align*}
    }
    \item {\rulename{$\Gamma$-Wt-Shared} : We know from premise that $p = * p_1$, $\Gamma\vdash p_1 \hookrightarrow \&^\kwd{s}(\tau)$, $\GWt{\Gamma}{p_1}{\&^\kwd{s}(\tau')}{\Gamma'}$, and $V\vDash p_1 \rightsquigarrow \&(\_, v)$, then we reach from hypothesis 
    \begin{align*}
        & \Phi(V:\Gamma')-\Phi(V:\Gamma) \\
        & = \phi(\&(\_, v): \&^\kwd{s}(\tau')) - \phi(\&(\_, v): \&^\kwd{s}(\tau)) \\
        & = \phi(v:\tau') - \phi(v:\tau)
    \end{align*}
    }
    \item {\rulename{$\Gamma$-Wt-Mutable} : We know from premise that $p = * p_1$, $\Gamma\vdash p_1 \hookrightarrow \&^\kwd{m}(\tau, \tau_\text{p})$, $ \GWt{\Gamma}{p_1}{\&^\kwd{m}(\tau', \tau_\text{p})}{\Gamma'}$, and $V\vDash p_1 \rightsquigarrow \&(\_, v)$, then we reach from hypothesis 
    \begin{align*}
        & \Phi(V:\Gamma')-\Phi(V:\Gamma) \\
        & = \phi(\&(\_, v): \&^\kwd{m}(\tau', \tau_\text{p})) - \phi(\&(\_, v): \&^\kwd{m}(\tau, \tau_\text{p})) \\
        & = [\phi(v:\tau')-\phi(v:\tau_\text{p}] - [\phi(v:\tau)-\phi(v:\tau_\text{p})] \\
        & = \phi(v:\tau') - \phi(v:\tau)
    \end{align*}
    }
\end{enumerate}
\end{proof}

\newpage
\begin{lemma}[Evaluation]
If $V\vDash e\rightsquigarrow v, \Gamma\vdash e\hookrightarrow \tau\dashv\Gamma'$, then $\Phi(V:\Gamma')-\Phi(V:\Gamma)= -\phi(v:\tau)$.
\end{lemma}
\begin{proof}
By induction on $e$:
\begin{enumerate}
    \item {$e=\kwd{n}_\text{i32}, e_1 ~\kwd{op}~ e_2, \kwd{true}, \kwd{false}, \kwd{copy}~ p, \kwd{nil}$ : $\Gamma'=\Gamma$ and $\phi(v:\tau) = 0$;}
    \item {$e=\kwd{move}~p$ : We know that $V\vDash p\rightsquigarrow v$, $\Gamma\vdash p\hookrightarrow \tau$ and $\GWt{\Gamma}{p}{\bot}{\Gamma'}$, hence $\Phi(V:\Gamma')-\Phi(V:\Gamma)=\Phi(v:\bot)-\Phi(v:\tau)=-\Phi(v:\tau)$;} 
    \item {$e=\kwd{box}(e_1)$: simply by induction, $\Phi(V:\Gamma')-\Phi(V:\Gamma)=\Phi(v_1:\tau_1)=\Phi(\kwd{box}(v_1):\kwd{box}(\tau_1))$;}
    \item {$e=\&^\kwd{s}~p$: We assert that for all $\textit{share}~\tau~\textit{as}~\tau_1, \tau_2$, $\phi(v:\tau)=\phi(v:\tau_1)+\phi(v:\tau_2)$, then $\Phi(V:\Gamma')-\Phi(V:\Gamma)=\phi(v:\tau)-\phi(v:\tau_1)=\phi(v:\tau_2)=\phi(\&(p, v):\&^\kwd{s}(\tau_2))$, where the assertion can be proved by simple induction on sharing;}
    \item {$e=\&^\kwd{m}~p$: $V\vDash p\rightsquigarrow v, \Gamma\vdash p\hookrightarrow \tau, \GWt{\Gamma}{p}{\tau_\text{p}}{\Gamma'}$, with $\textit{prophecy}~ \tau ~\textit{as}~ \tau_\text{p}$, then $\Phi(V:\Gamma')-\Phi(V:\Gamma)=\phi(v:\tau_\text{p})-\phi(v:\tau)= - \phi(\&(p, v):\&^\text{m}(\tau, \tau_\text{p}))$.}
\end{enumerate}
\end{proof}

\begin{theorem}[Soundness]
If $V\vDash s \rightsquigarrow^{\delta_V} \Dashv V', \Gamma\vdash s \hookrightarrow^{\delta_\Gamma}\dashv\Gamma'$, then $\Phi(V':\Gamma')-\Phi(V:\Gamma)+\delta_V\leq\delta_\Gamma$.
\end{theorem}
\begin{proof}
By induction on $V\vDash s \rightsquigarrow^{\delta_V} \Dashv V'$:
\begin{enumerate}
    \item { \rulename{V-Ex-Ret} $s=\kwd{return}$ : $\delta_V=\delta_\Gamma=0, V'=V, \Gamma'=\Gamma$;}
    \item { \rulename{V-Ex-Seq} $s=s_1; s_2$ : By induction, we have $V\vDash s_1\rightsquigarrow^{\delta_{V, 1}}\Dashv V_1\vDash s_2\rightsquigarrow^{\delta_{V, 2}}\Dashv V'$, $\Gamma\vdash s_1\hookrightarrow^{\delta_{\Gamma, 1}}\dashv \Gamma_1\vdash s_2\hookrightarrow^{\delta_{\Gamma, 1}}\dashv \Gamma'$ and $\Phi(V_1:\Gamma_1)-\Phi(V:\Gamma)+\delta_{V, 1}\leq\delta_{\Gamma, 1}, \Phi(V':\Gamma')-\Phi(V_1:\Gamma_1)+\delta_{V, 2}\leq\delta_{\Gamma, 2}$. It is obvious that $\Phi(V:\Gamma)-\Phi(V':\Gamma')+\delta_{\Gamma, 1}+\delta_{\Gamma, 2}\geq\delta_{V, 1}+\delta_{V, 2}$;}
    \item {\rulename{V-Ex-Tick} $s=\kwd{tick}(\delta)$ : $\delta_V=\delta_\Gamma=\delta, V'=V, \Gamma'=\Gamma$;}
    \item {\rulename{V-Ex-Drop} $s=\kwd{drop}~p$ : $V=V', \delta_V=\delta_\Gamma=0$, we only need to prove that $\Phi(V:\Gamma')-\Phi(V:\Gamma')\leq 0$, but we know that $\GWt{\Gamma}{p}{\bot}{\Gamma'}$, then $\Phi(V:\Gamma')-\Phi(V:\Gamma)=\phi(v:\bot)-\phi(v:\tau)=-\phi(v:\tau)\leq0$, with the help of $\vdash \tau$ from premise of $\Gamma \vdash s \hookrightarrow^{\delta_\Gamma}\dashv \Gamma'$;}
    
    \item {\rulename{V-Ex-Assign} $s=p\from e$ : $\delta_V=\delta_\Gamma=0$, and we have $\Gamma\vdash e\hookrightarrow\tau'\dashv\Gamma_1$, $\Gamma_1\vdash p\hookrightarrow\tau$, $\GWt{\Gamma_1}{p}{\tau'}{\Gamma'}$, $V\vDash e\rightsquigarrow v'$, $V\vDash p\rightsquigarrow v$ and $\VWt{V}{p}{v'}{V'}$. With lemmas, we can imply that $\Phi(V:\Gamma_1)-\Phi(V:\Gamma)=-\phi(v':\tau')$, $\Phi(V:\Gamma')-\Phi(V:\Gamma_1)=\phi(v:\tau')-\phi(v:\tau)$ and $\Phi(V':\Gamma')-\Phi(V:\Gamma')=\phi(v':\tau')-\phi(v:\tau')$. Therefore $\Phi(V':\Gamma')-\Phi(V:\Gamma)=-\phi(v:\tau)\leq0$, with the help of $\vdash \tau$ from premise of $\Gamma \vdash s \hookrightarrow^{\delta_\Gamma}\dashv \Gamma'$; }
    
    \item{\rulename{V-Ex-Cons} $s=p\from \kwd{cons}(e_1; e_2)$ : $\delta_V=0, \delta_\Gamma=\alpha'$, and we have $\Gamma\vdash e_1\hookrightarrow\kwd{i32}$, $\Gamma\vdash e_2\hookrightarrow\kwd{box}(\kwd{list}(\alpha'))\dashv\Gamma_1$, $\Gamma_1\vdash p\hookrightarrow\kwd{list}(\alpha)$, $\GWt{\Gamma_1}{p}{\kwd{list}(\alpha')}{\Gamma'}$, $V\vDash p\rightsquigarrow v$, $V\vDash e_1\rightsquigarrow \kwd{n}_\text{i32}$, $V\vDash e_2\rightsquigarrow \kwd{box}(lv)$ and $\VWt{V}{p}{\kwd{cons}(n, \kwd{box}(lv))}{V'}$. With lemmas, we can imply that $\Phi(V:\Gamma_2)-\Phi(V:\Gamma)=-\phi(\kwd{box}(lv):\kwd{box}(\kwd{list}(\alpha)))$, $\Phi(V:\Gamma')-\Phi(V:\Gamma_2)=\phi(v:\kwd{list}(\alpha'))-\phi(v:\kwd{list}(\alpha))$, $\Phi(V':\Gamma')-\Phi(V:\Gamma')=\phi(\kwd{cons}(\kwd{n}_\text{i32}, \kwd{box}(lv)):\kwd{list}(\alpha'))-\phi(v:\kwd{list}(\alpha'))$. Therefore, we have $\Phi(V':\Gamma')-\Phi(V:\Gamma)=\alpha-\phi(v:\kwd{list}(\alpha'))\leq\alpha=\delta_\Gamma$, with the help of $\vdash \kwd{list}(\alpha')$;}
    
    \item {\rulename{V-Ex-IfT/F}$s=\kwd{if}~p~\kwd{then}~s_1~\kwd{else}~s_2~\kwd{end}$ : If $V\vDash p\rightsquigarrow\kwd{true}$, then we have $V\vDash s_1\rightsquigarrow^{\delta_{V, 1}}\Dashv V'$, $\Gamma\vdash s_1\hookrightarrow^{\delta_{\Gamma, 1}}\dashv\Gamma_1$ and $\Phi(V':\Gamma_1)-\Phi(V:\Gamma)+{\delta_{V, 1}}\leq{\delta_{\Gamma, 1}}$. $\Gamma_1\sqcap\Gamma_2=\Gamma'$ indicates $\Phi(V':\Gamma')\leq\Phi(V':\Gamma_1)$. From premise of statics, $\delta_{\Gamma, 1}\leq\delta_\Gamma$, hence $\Phi(V':\Gamma')-\Phi(V:\Gamma)+\delta_{V, 1}\leq\delta_\Gamma$. If $V\vDash p\rightsquigarrow\kwd{false}$, it is similar to prove;}
    \item {\rulename{V-Ex-Mat-Nil/Cons}$s=\kwd{match}~p~\{\kwd{nil}\mapsto s_1, \kwd{cons}(x_\text{hd}, x_\text{tl})\mapsto s_2\}$ : If $V\vDash p\rightsquigarrow\kwd{nil}$, it is similar to $\kwd{if}$ statement. We now turn to the possibility $V\vDash p\rightsquigarrow\kwd{cons}(n, \kwd{box}(lv))$. In such a case, we can obviously get that $\Phi(V_\text{b}:\Gamma_\text{b})-\Phi(V:\Gamma)=-\alpha$. From $V_\text{b}\vDash s_2\rightsquigarrow^{\delta_V}\Dashv V'_\text{b}, \Gamma_\text{b}\vdash s_2\hookrightarrow^{\delta_{\Gamma, 2}}\dashv \Gamma'_\text{b}$, we have that $\Phi(V'_\text{b}:\Gamma'_\text{b})-\Phi(V_\text{b}:\Gamma_\text{b})+\delta_V\leq\delta_{\Gamma, 2}$. Also, $\Phi(V':\Gamma_2)-\Phi(V'_\text{b}:\Gamma'_\text{b})=\beta$. Sum up inequalities, we have $\Phi(V':\Gamma_2)-\Phi(V:\Gamma)+\delta_V\leq\delta_{\Gamma, 2}-\alpha+\beta$. Similarly with $\Gamma'=\Gamma_1\sqcap\Gamma_2$ and $\delta_{\Gamma, 2}-\alpha+\beta\leq\delta_\Gamma$, we final reach that $\Phi(V':\Gamma')-\Phi(V:\Gamma)+\delta_V\leq\delta_\Gamma$; }
    \item {\rulename{V-Ex-App} $s=p\from f(\vec{e})$ : Assume $\text{fn}~ f ~(\vec{x}_\text{param}:\vec{t}_\text{param}, \vec{x}_\text{local}:\vec{t}_\text{local}, x_\text{ret}:t_\text{ret}) \{~ s_f ~\}, V\vDash p\rightsquigarrow v$, $\Gamma_n\vdash p\hookrightarrow \tau, \vdash \tau$, $V\vDash p\from f(\vec{e})\rightsquigarrow^{\delta_V}\Dashv V'$, and $\Gamma\vdash p\from f(\vec{e})\hookrightarrow^{\delta_\Gamma}\dashv\Gamma'$.
    
    To use induction hypothesis on $V_\text{b}\vDash s_f\rightsquigarrow^{\delta_V} V'_\text{b}$ from premise of dynamics, we need statics $\Gamma_\text{b}\vdash s_f\hookrightarrow^{\delta_\text{b}}\dashv \Gamma'_\text{b}$, where $\GWt{\Gamma_n}{\vec{x}_\text{param}}{\vec{\tau}_\text{arg}}{\Gamma_{\text{b}, 1}}$, $\GWt{\Gamma_{\text{b}, 1}}{\vec{x}_\text{local}}{\vec{\tau}_\text{local}}{\Gamma_{\text{b}, 2}}$ and $\GWt{{\Gamma_{\text{b}, 2}}}{x_\text{ret}}{\tau_\text{ret}}{\Gamma_\text{b}}$. However, it is not found but $\Gamma_f\vdash s_f \hookrightarrow^\delta \dashv\Gamma'_f$ appears in premise of $\vdash f \Leftarrow (\Gamma_f, \delta_f)$. Via context extension and weakening rules, we can reach our goal by proving $\Gamma_f \sqsubseteq \Gamma_\text{b}$. It is true because $\forall i=1, ..., n, \tau_{\text{param}, i} = \tau_{\text{arg}, i}$.
    
    \begin{enumerate}
        \item {$\Phi(V:\Gamma_n)-\Phi(V:\Gamma) = -\sum_i\phi(v_i:\tau_{\text{arg}, i})$;}
        \item {$\Phi(V_\text{b}:\Gamma_\text{b}) - \Phi(V:\Gamma_n) = \sum_i\phi(v_i:\tau_{\text{arg}, i})$; }

        \item {$\Phi(V'_\text{b}:\Gamma'_\text{b}) - \Phi(V_\text{b}:\Gamma_\text{b}) + \delta_V \leq \delta_\text{b}$; }
        
        \item {$\Phi(V':\Gamma')-\Phi(V'_\text{b}:\Gamma'_\text{b}) = -\sum_{x\in\vec{x}_\text{param} \text{or} x\in\vec{x}_\text{local}}
        \phi(v_x:\tau_x)-\phi(v:\tau) \leq 0$, \\
        considering that it does not affect resource to move $v_\text{ret}:\tau'_\text{ret}$ from $x_\text{ret}$ to $p$, \\
        while the erasure at parameters, local variables and $p$ affects; }

        \item {$\delta_\text{b} = \delta$ by weakening rules;}
        \item {$\delta = \delta_f$ from premise of $\vdash f \Leftarrow (\Gamma_f, \delta_f)$; }
        \item {$\delta_f = \delta_\Gamma$ from statics.}
    \end{enumerate}
    Sum up inequalities above, we reach $\Phi(V':\Gamma')-\Phi(V:\Gamma)+\delta_V\leq\delta_\Gamma$.
    }
\end{enumerate}
\end{proof}

\end{document}